\documentclass[letter,10pt]{article}
\usepackage{fullpage}
\pdfoutput=1

\usepackage[cmex10]{amsmath}
\usepackage{wasysym} 
\usepackage{amsthm}
\usepackage{amsfonts}
\usepackage{ amssymb }
\usepackage{graphicx}
\usepackage{wrapfig}
\usepackage{subfig}
\usepackage[noadjust]{cite}
\usepackage{mathrsfs} 
\usepackage{centernot}
\usepackage[samesize]{cancel}
\usepackage{algpseudocode}
\usepackage{float}

\usepackage{marginnote} 
\usepackage{color}  

\floatstyle{ruled}
\newfloat{alg}{H}{loalg}
\floatname{alg}{Algorithm}
\floatstyle{plain}

\newtheorem{thm}{Theorem}
\newtheorem{cor}{Corollary}
\newtheorem{lem}{Lemma}
\newtheorem{prop}{Proposition}
\newtheorem{ex}{Example}
\newtheorem{defn}{Definition}
\newtheorem{prob}{Problem}
\newtheorem{rem}{Remark}
\newtheorem{ass}{Assumption}

\DeclareMathOperator*{\minimize}{min.}
\DeclareMathOperator{\st}{subj.\; to}

\DeclareMathOperator{\Boxed}{Box}

\DeclareMathOperator{\Viab}{Viab}
\DeclareMathOperator{\Disc}{Disc}
\DeclareMathOperator{\Inv}{Inv}

\DeclareMathOperator*{\Limsup}{Lim\,sup}
\DeclareMathOperator*{\Liminf}{Lim\,inf}

\newcommand{\tr}[1]{#1^\textsl{T}}
\newcommand{\s}{\mathcal{S}}

\newcommand{\Y}{\mathcal{Y}}
\newcommand{\D}{\mathcal{D}}
\newcommand{\V}{\mathcal{V}}

\newcommand{\C}{\mathcal{C}}

\newcommand{\B}{\mathcal{B}}

\newcommand{\U}{\mathcal{U}}

\newcommand{\Ul}{\mathscr{U}}
\newcommand{\Vl}{\mathscr{V}}

\newcommand{\Sl}{\mathbb{S}}

\newcommand{\X}{\mathcal{X}}
\newcommand{\K}{\mathcal{K}}
\newcommand{\A}{\mathcal{A}}
\newcommand{\Real}{\mathbb{R}}

\newcommand{\Lom}{\mathcal{L}}

\newcommand{\abs}[1]{\lvert#1\rvert}
\newcommand{\set}[1]{\left\{#1\right\}}

\newcommand{\norm}[1]{\left\lVert#1\right\rVert}
\newcommand{\normshort}[1]{\lVert#1\rVert}

\newcommand{\infnorm}[1]{\gamma}

\newcommand{\p}{\phantom{+}}

\begin{document}

\title{A Modified Riccati Transformation for Decentralized Computation of the Viability Kernel Under LTI Dynamics\thanks{Research supported by NSERC Discovery Grant \#327387 (M.~Oishi), NSERC Collaborative Health Research Project \#CHRPJ-350866-08 (G.~Dumont), and the Institute for Computing, Information and Cognitive Systems (ICICS). This work was mainly carried out at Electrical \& Computer Engineering, University of British Columbia, Vancouver, BC V6T 1Z4, Canada.}}

\author{Shahab Kaynama\thanks{S.\ Kaynama ({\tt\footnotesize kaynama@ece.ubc.ca}, cor.\ author) is currently with Electrical Engineering \& Computer Sciences, University of California at Berkeley, 337 Cory Hall, Berkeley, CA 94720, USA.} ~and Meeko Oishi\thanks{M.\ Oishi ({\tt\footnotesize oishi@unm.edu}) is with Electrical \& Computer Engineering, University of New Mexico, MSC01 1100, 1 University of New Mexico, Albuquerque, NM 87131, USA.}}

\date{(Preprint Submitted for Publication)}

\maketitle

\begin{abstract}
    Computing the viability kernel is key in providing guarantees of safety and proving existence of safety-preserving controllers for constrained dynamical systems. Current numerical techniques that approximate this construct suffer from a complexity that is exponential in the dimension of the state. We study conditions under which a linear time-invariant (LTI) system can be suitably decomposed into lower-dimensional subsystems so as to admit a conservative computation of the viability kernel in a decentralized fashion in subspaces. We then present an isomorphism that imposes these desired conditions, particularly on two-time-scale systems. Decentralized computations are performed in the transformed coordinates, yielding a conservative approximation of the viability kernel in the original state space. Significant reduction of complexity can be achieved, allowing the previously inapplicable tools to be employed for treatment of higher-dimensional systems. We show the results on two examples including a 6D system.
\end{abstract}





\section{Introduction} \label{S:Intro}

Constrained dynamical systems have received a tremendous amount of attention due to the presence of safety constraints and hard bounds that appear in many practical scenarios. Providing guarantees of constraint satisfaction and facilitating synthesis of constraint-satisfying controllers therefore is highly desirable, particularly in safety-critical applications. A class of safety-critical systems known as \emph{envelope protection problems} is concerned with ensuring that the trajectories remain in a safe, bounded ``envelope'' (subset) of the state space for a given time horizon. Such problems arise in e.g.\ flight management systems \cite{Margellos_Lygeros_2009,Lygeros1999,Tomlin2000a,TMBO03} where the safety constraints are defined as the aircraft's aerodynamic envelope and consequently the system must ensure that certain combinations of states are avoided to prevent stalling or other undesirable behaviors. Other application domains include control of depth of anesthesia \cite{kaynama_HSCC2012}, aircraft autolanders \cite{BMOT07}, automated highway systems \cite{Lygeros_Godbole_Sastry_1998}, control of under-actuated underwater vehicles \cite{Panagou_Margellos_Summers_Lygeros_Kyriakopoulos_2009}, stockout prevention of storage systems in manufacturing processes \cite{Borrelli_Del_Vecchio_Parisio_2009}, and management of a marine renewable resource \cite{Bene2001}, to name a few. 

Viability theory \cite{KurzhanskiFillipov87, aubin1991viability, blanchini2008set} provides a set-valued perspective on the behavior of the trajectories inside a given set. Thus it is naturally suited to handle envelope protection problems. By duality, \emph{minimal} reachability \cite{MBT05} is also capable of analyzing such problems by investigating the behavior of the trajectories outside of the envelope. For simplicity, in this paper we only focus on the constructs generated within the framework of viability theory. The \emph{viability kernel} is the set of initial states for which there exists at least one trajectory of the input-constrained system that respects the state constraint for all time. It is shown in \cite{aubin1991viability} and (by duality in \cite{MitchellHSCC07}) that the viability kernel is the \emph{only} construct that can be used to prove safety/viability of the system and to synthesize inputs that preserve this safety; cf.\ \cite[Chap.\ 1--2]{kaynama2012Thesis} for more detail. In general an exact computation of the viability kernel is extremely difficult if not impossible. Instead, approximations of this set are computed. Such computations have historically been subject to Bellman's ``curse of dimensionality'' \cite{ADFGLM06}. The numerical algorithms that approximate the viability kernel and its associated control laws (e.g., \cite{MBT05, Saint-Pierre_1994, cardaliaguet1999set, gao2006-reachability}), collectively referred to as \emph{Eulerian methods} \cite{MitchellHSCC07}, rely on gridding the state space and therefore their computational complexity increases exponentially with the dimension of the state. This renders them impractical for systems of dimension higher than three or four.

This paper presents a part of our efforts to address the curse of dimensionality by enabling the use of Eulerian algorithms for higher-dimensional LTI systems (and by extension, hybrid systems with LTI dynamics). We decompose the structure of the system, applying Eulerian algorithms on each individual lower-dimensional subsystem in a decentralized fashion. Significant computational gains can be obtained, since instead of one costly centralized computation on the full-order system, multiple less expensive subsystem computations are performed. The results are then mapped back to the full-order space to obtain a conservative approximation (i.e.\ an \emph{under-approximation}) of the viability kernel. The contribution of this paper is twofold: 1) We investigate various structures on system matrices that must be satisfied so that the behavior of the constrained system for envelope protection problems (with simply-connected, compact constraints) can be inferred conservatively from subspace decentralized analyses (Section~\ref{S:Decentralized_Viab}). 2) We then present an isomorphism through which the desired structure is \emph{imposed} on the system (albeit under certain conditions) to facilitate decentralized computations in the transformed space (Section~\ref{S:Riccati}). Numerical examples are provided in Section~\ref{S:Riccati_examples}.

\subsection{Related Work}

Complexity reduction for viability and minimal reachability has been addressed by many researchers. A projection scheme in \cite{MT03} based on Hamilton-Jacobi (HJ) partial differential equations (PDEs) over-approximates the projection of the true minimal reachable tube in lower dimensional subspaces, with the unmodeled dimensions treated as a disturbance. Similarly, \cite{SHT03} decomposes a full-order nonlinear system into either disjoint or overlapping subsystems and solves multiple HJ PDEs in lower dimensions. More recently, a mixed implicit-explicit HJ is presented in \cite{Mitchell_HSCC11} for nonlinear systems whose state vector contains states that are integrators of other states. The complexity of this new formulation is linear in the number of integrator states, while still exponential in the dimension of the rest of the states. These techniques assume that the system itself presents a certain structure that can be exploited.

In \cite{Coquelin_Martin_Munos_2007}, an approximate dynamic programming technique is presented that, although still grid-based, enables a more efficient computation of the viability kernel. The viability kernel (similarly to \cite{Lygeros2004}) is expressed as the zero sublevel set of the value function of the corresponding optimal control problem. It is assumed that the value function, which is a viscosity solution of a HJB PDE, is differentiable everywhere on the constraint set. The PDE is then discretized and the resulting value function is numerically computed on a grid using a function approximator such as the $k$-nearest neighbor algorithm. The error-bounded approximation is not conservative (it is an over-approximation) but converges to the true viability kernel in the limit as the number of grid points goes to infinity.


Another related approach is the search for a barrier certificate \cite{Prajna_Jadbabaie_2004}, a Lyapunov-like function that forms a separating hyper-surface between any two given sets $\A$ and $\B$ in the state space. If there exists a function non-positive on $\A$ and positive on $\B$, and whose Lie derivative (along the vector field) is non-positive on its zero level set for all states and controls, then no trajectories will ever go from $\A$ to $\B$. This technique can be adapted to analytically describe the boundary of the \emph{infinite-horizon} viability kernel: A certificate must now be formulated such that at every state along its zero level set there exists a control that makes the Lie derivative non-positive. For systems with polynomial vector fields and semi-algebraic constraints, efficient techniques based on Sum of Squares can be used to find the barrier certificate.\footnote{This method cannot be used to formulate the \emph{finite-horizon} viability kernel which may be useful when, for example, the infinite-horizon kernel is empty, or when safety is to be verified/enforced over a finite time interval. Moreover, there are no guarantees that a barrier certificate can be found for a given system no matter how simple its dynamics (even when a Lyapunov function is already known).}

Recently, we presented a connection between the viability kernel and efficiently-computable classes of reachability constructs known as maximal reachable sets. Owing to this connection, scalable numerical algorithms (collectively referred to as \emph{Lagrangian methods} \cite{MitchellHSCC07}) such as \cite{Le_Guernic_Girard_2010,SpaceEx_2011,Kurzhanski2000a, KV06, GGM06, Girard2008, Han_Krogh_2006}, originally developed for maximal reachability, can now be used to approximate the viability kernel. We presented two algorithms for LTI systems with convex constraints based on piecewise ellipsoidal representations \cite{kaynama_HSCC2012} and support vectors \cite{kaynama_Aut2012} that have polynomial complexity. In contrast to these results, the technique presented here reduces the complexity indirectly by decentralizing computations. The benefit of this approach is that it allows useful features of Eulerian methods such gradient-based control synthesis and handling of arbitrarily shaped nonconvex constraints be taken advantage of.


\section{Problem Statement} \label{S:problem_formulation}

Consider the continuous-time system
\begin{equation}\label{E:nonlinear_ss_eqn}
    \dot{x}=f(x,u)
\end{equation}
with state space $\X:=\Real^n$ (a finite-dimensional vector space), state vector $x(t)\in\X$, and input $u(t)\in\U$ where $\U$ is a compact (closed and bounded) and convex subset of $\Real^p$. The vector field $f\colon \X \times \U \to \X$ is assumed to be Lipschitz in $x$ and continuous in $u$. Let
\begin{equation}\label{E:U_measurable}
    \Ul_{[0,t]}:=\left\{u\colon [0,t] \to \Real^p \; \text{measurable}, \;\; u(s)\in \U \;\,\text{a.e.}\; s\in [0,t] \right\}.
\end{equation}
With an arbitrary, finite time horizon $\tau>0$, for every $t\in [0,\tau]$, $x_0\in \X$, and $u(\cdot)\in \Ul_{[0,t]}$, there exists a unique trajectory $\xi_{x_0,u}\colon [0,t] \to \X$ that satisfies \eqref{E:nonlinear_ss_eqn} and the initial condition $\xi_{x_0,u}(0)=x_0$.

For a nonempty, simply-connected, compact state constraint set $\K \subset \X$ we are concerned with computing the following backward construct:\footnote{By duality the arguments presented in this paper also hold for the minimal reachable tube of $\K^c$; cf.\ \cite{kaynama2012Thesis}.}
\begin{defn}[Viability Kernel] \label{D:viab}
    The finite-horizon viability kernel\footnote{The infinite-horizon viability kernel $\Viab_{\Real^+}(\K,\U)$ is also known as the \emph{maximal controlled-invariant set} \cite{Blanchini_1999}.} of $\K$ is the set of initial states for which there exists an input such that the trajectories emanating from those states remain in $\K$ for all time $t\in[0,\tau]$:
    \begin{equation*}
            \Viab_{[0,\tau]}(\K,\U):= \left\{ x_0 \in \K \mid  \exists  u(\cdot)\in \Ul_{[0,\tau]}, \,  \forall  t\in [0,\tau],\, \xi_{x_0,u}(t) \in \K \right\}.
    \end{equation*}
\end{defn}

Initial states belonging to this set are viable under \eqref{E:nonlinear_ss_eqn}, and the corresponding control laws are safety-preserving. The powerful Eulerian methods are capable of directly computing the viability kernel and its safety-preserving control policies. However, they rely on gridding the state space, and therefore are computationally intensive. Although versatile in terms of ability to handle various types of dynamics and constraints, the applicability of these techniques has been historically limited to systems of low dimensionality (up to 4D in practice) due to their exponential complexity.


We restrict ourselves to LTI systems of the form
\begin{equation} \label{E:general_ss_eqn}
    \dot{x}=Ax+Bu
\end{equation}
described by the matrix notation
\begin{equation} \label{E:general_LTI_system}
\s := \begin{bmatrix}
    \begin{array}{c|c}
       A &  B  \\
    \end{array} \\
    \end{bmatrix}
\end{equation}
with constant, appropriately sized $A$ and $B$ matrices.

\begin{prob}[Decentralized Viability]\label{Prob:problem_formulation}
    i) Identify a structure on $A$ and $B$ for which the viability kernel can be conservatively reconstructed from its subsystem analyses. ii) Find an isomorphic state space for \eqref{E:general_ss_eqn} in which the system has this desired structure.
\end{prob}

\subsection{Preliminaries}


\paragraph{Notation}
For a set $\A\subseteq \X$, $\A^c$ and $2^\A$ denote the complement and the power set of $\A$ in $\X$, respectively. For brevity, $\norm{\cdot}$ denotes the infinity norm. For a constant matrix $A = \left[a_{ij}\right]\in \Real^{m \times n}$ the induced norm is $\norm{A} := \sup_{v \in \Real^n,\, v \neq 0} \frac{\norm{Av}}{\norm{v}} = \max_{1 \leq j \leq n} \sum_{i=1}^m \abs{a_{ij}}$. For a Lebesgue measurable function $f\colon \Real \to \Real^n$ defined over an interval $[t_a,t_b]$ we denote $\norm{f}:=\norm{f(\cdot)}_{\Lom_\infty [t_a,t_b]} = \sup_{t\in [t_a,t_b]} \norm{f(t)} < \infty$. A linear transformation of $\s$ in \eqref{E:general_LTI_system} using a nonsingular matrix $T \in \Real^{n \times n}$ is defined as $\s' = T^{-1}(\s) := \begin{bmatrix}
\begin{array}{c|c}
       T^{-1}AT &  T^{-1}B
\end{array}
\end{bmatrix}$. A linear transformation of a set $\A \subseteq \X$ under the same mapping is $\Y = T^{-1} \A := \{ y \mid y = T^{-1} a, \: a \in \A \}$.

\begin{defn}[Disjoint Input]\label{Def:disjoint_input}
    The input $u =\tr{\left[u_1 \dotsb u_p\right]} \in \U \subset \Real^p$ is \emph{disjoint} across two subsystems
    \begin{subequations}
    \begin{align}\label{E:subsys_Defn_disjoint}
        \dot{x}_1 &= A_1 x_1 + \Delta_{12} x_2 +  B_1 u,\\
        \dot{x}_2 &= A_2 x_2 + \Delta_{21} x_1 + B_2 u
    \end{align}
    \end{subequations}
    of an LTI system with $x_1 \in \Real^k$ and $x_2 \in \Real^{n-k}$ if $\forall s\in \{1,\dots,p\}$, $i\neq j$,
    \begin{equation}
        \frac{\partial B_i u}{\partial u_s} \neq 0 \, \rightarrow \, \frac{\partial B_j u}{\partial u_s}=0,
    \end{equation}
    and $\U = \U_1 \times \U_2$, where $\U_i$ is any (possibly degenerate) subset of $\Real^p$ from which the portion of the vector $u$ acting directly on subsystem $i$ draws its values. 
\end{defn}

\begin{defn}[Unidirectionally Coupled] \label{D:UnidirectionalCouple}
    The subsystems
    \begin{subequations}
    \begin{align}
        \dot{x}_1 &= A_1 x_1 + B_1 u, \label{E:I0}\\
        \dot{x}_2 &= A_2 x_2 + \Delta_{21} x_1 + B_2 u \label{E:I}
    \end{align}
    \end{subequations}
    with disjoint input across them are said to be \emph{unidirectionally coupled} since the trajectories of \eqref{E:I} are affected by those of \eqref{E:I0}, while \eqref{E:I0} evolves independently from \eqref{E:I}. The worst-case unidirectional coupling can be characterized by $\norm{\Delta_{21}}$.
\end{defn}

\begin{defn}[ETUC] \label{D:trivial-unctrl}
    A subsystem is said to be \emph{externally trivially uncontrollable} (ETUC) if it possesses a null input matrix.
\end{defn}

\begin{rem}\label{rem:shape_of_U}
    The condition on $\U$ in Definition~\ref{Def:disjoint_input} enures that the inputs acting on each subsystems are independent of one another. This condition is satisfied for most physical systems where actuators are commonly uncorrelated, or for a system with an ETUC subsystem (in which case the shape of $\U$ becomes irrelevant). In the most general case, however, $\U$ can be (under-)approximated by a cross-product set. 
\end{rem}

\section{Decentralized Viability Computation}\label{S:Decentralized_Viab}

We begin by arriving at the desired structure on system matrices that would allow for decentralized (and conservative) computation of the viability kernel. Throughout the paper we assume a partitioning of \eqref{E:general_LTI_system} that results in two subsystems. The arguments can be easily generalized to multiple subsystems as discussed in Section~\ref{S:recursive_decomp}.


\subsection{Why Decoupling of $A$ Alone is Insufficient}

Consider the following system with block diagonal $A$-matrix, and a $B$-matrix of generic form:
\begin{equation}\label{E:sys_decoupled_genericB}
    \begin{bmatrix}
        \dot{x}_1\\
        \dot{x}_2
    \end{bmatrix}
    =\begin{bmatrix}
        A_1 & \mathbf{0}\\
        \mathbf{0} & A_2
    \end{bmatrix}
    \begin{bmatrix}
        x_1\\
        x_2
    \end{bmatrix}
    +
    \begin{bmatrix}
        B_1\\
        B_2
    \end{bmatrix}u, \quad u\in \U.
\end{equation}
Denote the two subspaces of $\Real^n$ in which the subsystems evolve as
\begin{equation}
    \Sl_1:=\Real^k \quad \text{and} \quad \Sl_2:=\Real^{n-k}.
\end{equation}
Let $\Pi_i x$ be the projection of the vector $x = \tr{[x_1 \; x_2 ]} \in \X$ onto $\Sl_i$:
\begin{equation}
    \Pi_i x = x_i \in \Sl_i,
\end{equation}
and $\Pi_i \K$ the projection of the set $\K\subset \X$ onto $\Sl_i$:
\begin{equation}
    \Pi_i \K = \left\{ x_i \in \Sl_i \mid \exists x\in \K, \, \Pi_i x = x_i \right\}.
\end{equation}

\begin{lem}
    For any $t$ and $u(\cdot)\in \Ul_{[0,t]}$ the projection of trajectory $\xi$ of system \eqref{E:sys_decoupled_genericB} with initial condition $\xi_{x_0,u}(0)=x_0$ is a subsystem trajectory $\xi^i$ initiating from the projection of $x_0$:
    \begin{equation}
        \Pi_i \xi_{x_0,u}(t) = \xi^i_{\Pi_i x_0, u}(t).
    \end{equation}
\end{lem}


\begin{proof}
    $\Pi_i
    \left[\begin{smallmatrix}
            \dot{x}_1\\
            \dot{x}_2
        \end{smallmatrix}\right]
        =
        \Pi_i\left(
        \left[\begin{smallmatrix}
            A_1 & \mathbf{0}\\
            \mathbf{0} & A_2
        \end{smallmatrix}\right]
        \left[\begin{smallmatrix}
            x_1\\
            x_2
        \end{smallmatrix}\right]
        +
        \left[\begin{smallmatrix}
            B_1\\
            B_2
        \end{smallmatrix}\right]u
        \right)
        =
        A_i \Pi_i x_i + B_i u$.
\end{proof}

\begin{cor}\label{Cor:proj_traj_implies_traj_proj}
    %
    \begin{equation}
        \xi_{x_0,u}(t) \in \K \Rightarrow \xi^i_{\Pi_i x_0,u}(t) \in \Pi_i \K.
    \end{equation}
\end{cor}

Later we will show and utilize the fact that under certain conditions this implication is bidirectional.

\begin{prop}[Wrong Approximation]\label{Prop:viab_over-approx_GenericB}
    For dynamics \eqref{E:sys_decoupled_genericB} the cross-product of subsystem viability kernels of projections of $\K$ is a superset of the viability kernel of $\K$: 
    \begin{equation}
        \Viab_{[0,\tau]}(\K,\U) \subseteq \Viab_{[0,\tau]}(\Pi_1 \K, \U) \times \Viab_{[0,\tau]}(\Pi_2 \K,\U).
    \end{equation}
\end{prop}
\begin{proof}
    \begin{align*}
        x_0 \in \Viab_{[0,\tau]}(\K,\U) &\Leftrightarrow \exists u(\cdot),\, \forall t,\, \xi_{x_0,u}(t) \in \K\\
        &\Rightarrow \exists u(\cdot),\, \forall t,\, \left(\xi^1_{\Pi_1 x_0,u}(t) \in \Pi_1 \K \; \wedge \; \xi^2_{\Pi_2 x_0,u}(t) \in \Pi_2 \K\right)\\
        &\Rightarrow \exists u(\cdot),\, \forall t,\, \xi^1_{\Pi_1 x_0,u}(t) \in \Pi_1 \K \; \wedge \; \exists u(\cdot),\, \forall t,\, \xi^2_{\Pi_2 x_0,u}(t) \in \Pi_2 \K\\
        &\Rightarrow \Pi_1 x_0 \in \Viab_{[0,\tau]}(\Pi_1 \K,\U) \; \wedge \; \Pi_2 x_0 \in \Viab_{[0,\tau]}(\Pi_2 \K,\U)\\
        &\Rightarrow x_0 \in  \Viab_{[0,\tau]}(\Pi_1 \K,\U) \times \Viab_{[0,\tau]}(\Pi_2 \K,\U).
    \end{align*}
\end{proof}

The following counter example demonstrates that an inclusion in the opposite direction does not hold for system \eqref{E:sys_decoupled_genericB}; That is, $\Viab_{[0,\tau]}(\K,\U) \not\supseteq \Viab_{[0,\tau]}(\Pi_1 \K, \U) \times \Viab_{[0,\tau]}(\Pi_2 \K,\U)$. Consider the point $x'=\left[\begin{smallmatrix}1\\ 1\end{smallmatrix}\right]$ and constraint set $\K=[-1,1]\times[-1,1]$. We seek to compute the viability kernel of this set under the dynamics $\dot{x}_1 = x_1 + u$ and $\dot{x}_2 = x_2 - u$ and input constraint $u\in [-1,1]$. The point $x'$ belongs to the cross-product of subsystem viability kernels (since subsystem 1 can use $u=-1$ while subsystem 2 can use $u=+1$ at the same point to keep $\Pi_i x'$ in $\Pi_i \K$), but does not belong to the actual full-order kernel (since no input exists that can keep the system in $\K$).  As such, when the system is in the form of \eqref{E:sys_decoupled_genericB} performing the analysis on subsystems would yield an over-approximation of the viability kernel. This stems from the fact that the input is non-disjoint across the subsystems. On the other hand, we do have the following correct inclusion even with a non-disjoint input.

\begin{lem}\label{Lem:ill-under-approximation}
    The following holds for system \eqref{E:sys_decoupled_genericB}:
    \begin{equation}\label{E:incl_generic_1}
        \Viab_{[0,\tau]}(\K,\U) \supseteq \left( \Viab_{[0,\tau]}((\Pi_1 \K^c)^c, \U) \times \Sl_2 \right) \cup \left( \Sl_1 \times \Viab_{[0,\tau]}((\Pi_2 \K^c)^c,\U)\right).
    \end{equation}
\end{lem}

\begin{proof}
    \begin{align*}
        x_0\in &\left( \Viab_{[0,\tau]}((\Pi_1 \K^c)^c, \U) \times \Sl_2 \right) \cup \left( \Sl_1 \times \Viab_{[0,\tau]}((\Pi_2 \K^c)^c,\U)\right)\\
        &\Leftrightarrow \exists u(\cdot),\, \forall t, \, \xi^1_{\Pi_1 x_0,u}(t) \in (\Pi_1 \K^c)^c \; \vee \; \exists u(\cdot),\, \forall t, \, \xi^2_{\Pi_2 x_0,u}(t) \in (\Pi_2 \K^c)^c\\
        &\Leftrightarrow \left(\forall u(\cdot),\, \exists t,\, \xi^1_{\Pi_1 x_0,u}(t) \in \Pi_1 \K^c \; \wedge \; \forall u(\cdot),\, \exists t,\, \xi^2_{\Pi_2 x_0,u}(t) \in \Pi_2 \K^c \right)^c\\
        &\Rightarrow \left(\forall u(\cdot),\, \exists t,\, \left( \xi^1_{\Pi_1 x_0,u}(t) \in \Pi_1 \K^c \; \wedge \; \xi^2_{\Pi_2 x_0,u}(t) \in \Pi_2 \K^c \right) \right)^c\\
        &\Rightarrow \left(\forall u(\cdot),\, \exists t,\, \xi_{x_0,u}(t) \in \K^c \right)^c\\
        &\Rightarrow \exists u(\cdot),\, \forall t,\, \xi_{x_0,u}(t) \in \K\\
        &\Rightarrow x_0 \in \Viab_{[0,\tau]}(\K,\U).
    \end{align*}
\end{proof}

\begin{defn}[Ill-Posedness]
    We say that a viability problem is ill-posed if the state constraint is empty.
\end{defn}

\begin{prop}[Ill-Posed Approximation]\label{Prop:meaningless_over-approx_GenericB}
    When $\K$ is a bounded subset of $\X$ (which is the case in most envelope protection problems) the approximation in Lemma~\ref{Lem:ill-under-approximation} is ill-posed.
\end{prop}

The proof should be clear from the fact that for any bounded set $\K$ we have $(\Pi_i \K^c)^c = \emptyset$. 

\subsection{Suitable Structures for Decomposition}\label{S:suitable_structures}

Consider a system with block-diagonal $A$-matrix and a $B$-matrix that ensures a disjoint input across the subsystems, for instance
\begin{equation}\label{E:sys_decoupled_DisjointB}
    \begin{bmatrix}
        \dot{x}_1\\
        \dot{x}_2
    \end{bmatrix}
    =\begin{bmatrix}
        A_1 & \mathbf{0}\\
        \mathbf{0} & A_2
    \end{bmatrix}
    \begin{bmatrix}
        x_1\\
        x_2
    \end{bmatrix}
    +
    \begin{bmatrix}
        B_1 & \mathbf{0}\\
        \mathbf{0} & B_2
    \end{bmatrix}u, \quad
    u= \begin{bmatrix}
        u_1\\
        u_2
    \end{bmatrix}\in \U,
\end{equation}
when $\U = \U_1 \times \U_2$.

\begin{ass}\label{Ass:K_axis-aligned}
    The set $\K$ is a cross-product of two (arbitrarily-shaped) sets in $\Sl_i$.
\end{ass}

\begin{cor}
    Under Assumption~\ref{Ass:K_axis-aligned} the projection of a trajectory is contained in a set if and only if the subsystem trajectories are contained in the projection of the set:
    \begin{equation}
        \xi_{x_0,u}(t) \in \K \Leftrightarrow \xi^i_{\Pi_i x_0,u_i}(t) \in \Pi_i \K.
    \end{equation}
\end{cor}

\begin{thm}\label{Thm:Viab_equals_projectionViabs}
    The viability kernel of $\K$ under \eqref{E:sys_decoupled_DisjointB} can be computed exactly using subsystem kernels:
    \begin{align}
        \Viab_{[0,\tau]}(\K,\U) &= \Viab_{[0,\tau]}(\Pi_1 \K, \U_1) \times \Viab_{[0,\tau]}(\Pi_2 \K,\U_2).\label{E:incl_decoupled_1}
    \end{align}
\end{thm}

\begin{proof}
    \begin{align*}
        x_0 \in \Viab_{[0,\tau]}(\K,\U) &\Leftrightarrow  \exists u(\cdot),\, \forall t, \, \xi_{x_0,u}(t)\in \K\\
        &\Leftrightarrow \exists u(\cdot),\, \forall t,\, \left( \xi^1_{\Pi_1 x_0,u}(t) \in \Pi_1 \K \; \wedge \; \xi^2_{\Pi_2 x_0,u}(t) \in \Pi_2 \K \right) &&\text{(by Assumption~\ref{Ass:K_axis-aligned})}\\
        &\Leftrightarrow \exists u_1(\cdot),\, \forall t,\, \xi^1_{\Pi_1 x_0,u_1}(t) \in \Pi_1 \K \; \wedge \; \exists u_2(\cdot),\, \forall t,\, \xi^2_{\Pi_2 x_0,u_2}(t) \in \Pi_2 \K &&\text{(via disjoint input)}\\
        &\Leftrightarrow \Pi_1 x_0 \in \Viab_{[0,\tau]}(\Pi_1 \K, \U_1) \; \wedge \; \Pi_2 x_0 \in \Viab_{[0,\tau]}(\Pi_2 \K,\U_2)\\
        &\Leftrightarrow x_0 \in \Viab_{[0,\tau]}(\Pi_1 \K, \U_1) \times \Viab_{[0,\tau]}(\Pi_2 \K,\U_2).
    \end{align*}


\end{proof}






\begin{rem}
    The use of any decomposition technique for correct (conservative) approximation of the viability kernel is contingent on satisfaction of Assumption~\ref{Ass:K_axis-aligned} as shown previously. When $\K$ does not satisfy this assumption, it can be under-approximated by the union of direct-product sets. The viability kernel can be computed for each set separately in lower dimensions (which increases the computational complexity only linearly in the number of sets). The union of the resulting kernels in full dimensions under-approximates the true viability kernel. Parallelization of viability calculations in each subspace could further reduce the computational time.
\end{rem}




In general, we may not be able to simultaneously obtain a decoupled $A$-matrix and a disjoint input. Instead, suppose that the system is of the form
%
\begin{equation}\label{E:sys_unicoupled_DisjointB}
    \begin{bmatrix}
        \dot{x}_1\\
        \dot{x}_2
    \end{bmatrix}
    =\begin{bmatrix}
        A_1 & \mathbf{0}\\
        \Delta & A_2
    \end{bmatrix}
    \begin{bmatrix}
        x_1\\
        x_2
    \end{bmatrix}
    +
    \begin{bmatrix}
        B_1\\
        \mathbf{0}
    \end{bmatrix}u, \quad u\in \U
\end{equation}
which automatically ensures that the input $u$ is disjoint across the subsystems regardless of the shape of $\U$ since one of the two (unidirectionally coupled) subsystems is ETUC (Remark~\ref{rem:shape_of_U}). This system can be rewritten as
\begin{equation}\label{E:sys_unicoupled_DisjointB_disturbance}
    \begin{bmatrix}
        \dot{x}_1\\
        \dot{x}_2
    \end{bmatrix}
    =\begin{bmatrix}
        A_1 & \mathbf{0}\\
        \mathbf{0} & A_2
    \end{bmatrix}
    \begin{bmatrix}
        x_1\\
        x_2
    \end{bmatrix}
    +
    \begin{bmatrix}
        B_1\\
        \mathbf{0}
    \end{bmatrix}u
    +
    \begin{bmatrix}
        \mathbf{0}\\
        \Delta
    \end{bmatrix} x_1, \quad u\in \U.
\end{equation}
The evolution of $x_1$ is completely independent of the evolution of $x_2$. Its effect on the lower subsystem, mapped through $\Delta$, can be viewed as an exogenous input to the lower subsystem, that takes values on the (possibly time-varying) subset $\V(\cdot)$ of the upper subspace $\Sl_1$. Treating this additional input in the worst-case fashion results in conservatism. Hence, define the following construct:

\begin{defn}[Discriminating Kernel]\label{Defn:disc}
    Consider a system with adversarial inputs: control $u(t) \in \U$ and disturbance $v(t) \in \V(t)$, where $\V \colon [0,\tau] \to 2^{\Real^{p_v}}$ is a point-wise convex and compact set-valued map from $[0,\tau]$ to $\Real^{p_v}$. Let
    \begin{equation*}
        \Vl_{[0,t]}:=\{v \colon [0,t] \to \Real^{p_v} \;\text{measurable}, \;\, v(s)\in \V(s) \;\text{a.e.} \; s\in[0,t] \}.
    \end{equation*}
    To be conservative, we assume non-anticipative strategies $\rho$ for one of the inputs.\footnote{A map $\rho \colon \Vl_{[0,t]} \to \Ul_{[0,t]}$ is non-anticipative for $u$ if for every $v(\cdot), v'(\cdot) \in \Vl_{[0,t]}$, $v(s) = v'(s)$ implies $\rho[v](s) = \rho[v'](s)$ a.e.\ $s \in [0,t]$ \cite{Evans_Souganidis_1984}. Note that for linear systems the Isaac's condition holds \cite{MBT05}, and therefore it does not matter which input is selected to play with non-anticipative policies.} The finite-horizon discriminating kernel of $\K$ is the set of initial states for which there exists a control such that the trajectories emanating from those states remain in $\K$ for every disturbance for all time $t\in[0,\tau]$:
    \begin{equation*}\label{E:disc_setbuilder}
        \Disc_{[0,\tau]}(\K,\U,\V(\cdot)):=\bigl\{ x_0 \in \K \mid  \exists  \rho \colon \Vl_{[0,\tau]} \to \Ul_{[0,\tau]}, \, \forall v(\cdot)\in \Vl_{[0,\tau]}, \, \forall  t\in [0,\tau], \, \xi_{x_0,\rho[v],v}(t) \in \K \bigr\}.
    \end{equation*}
\end{defn}


We will use a ``$*$'' subscript to distinguish a construct formed under \eqref{E:sys_unicoupled_DisjointB_disturbance} when $x_1$ for the lower subsystem is treated as an adversarial disturbance. 

\begin{lem}\label{Lem:disc_subset_viab}
    The viability kernel of a set $\K$ under \eqref{E:sys_unicoupled_DisjointB_disturbance} is a superset of the discriminating kernel of $\K$ when $x_1$ is treated as a worst-case disturbance (assumed to draw values from some time-varying set $\V(\cdot)$ point-wise convex and compact in $\Sl_1$) to the lower subsystem:
    \begin{equation}
        \Viab_{[0,\tau]}(\K,\U) \supseteq \Disc_{[0,\tau]}(\K,\U,\V(\cdot))_*.
    \end{equation}
\end{lem}

\begin{proof}
    Let $\hat{\xi}$ denote the trajectory of the system when $x_1$ is treated as a disturbance to the lower subsystem.
    \begin{align*}
        x_0 \in \Disc_{[0,\tau]}(\K,\U,\V(\cdot))_* &\Leftrightarrow \exists \rho[v](\cdot),\, \forall v(\cdot), \, \forall t, \, \hat{\xi}_{x_0,\rho[v],v}(t) \in \K\\
        & \Rightarrow \exists u(\cdot), \, \forall t,\, \hat{\xi}_{x_0,u,v(t)=x_1(t)}(t) \in \K &&\text{(a specific disturbance)}\\
        & \Rightarrow \exists u(\cdot), \, \forall t,\, \xi_{x_0,u}(t) \in \K\\
        & \Rightarrow x_0 \in \Viab_{[0,\tau]}(\K,\U).
    \end{align*}
\end{proof}

\begin{defn}[Invariance Kernel]
    Consider a system with a disturbance input $v(t) \in \V(t)$ as its only input, where $\V(\cdot)$ is defined as in Definition~\ref{Defn:disc}. The finite-horizon invariance kernel of a set $\K$ is the set of initial states that remain in $\K$ for every disturbance for all time $t\in[0,\tau]$:
    \begin{align*}
        \Inv_{[0,\tau]}(\K,\V(\cdot)) := \{ x_0 \in \K \mid \forall v(\cdot)\in \Vl_{[0,\tau]}, \, \forall t \in [0,\tau], \, \xi_{x_0,v}(t) \in \K \}.
    \end{align*}
\end{defn}

\begin{thm}[Main Decentralization Result]\label{Thm:Viab_supset_projectionViabs}
    The viability kernel of a set $\K$ under \eqref{E:sys_unicoupled_DisjointB} can be conservatively approximated using the subsystem viability/invariance kernels as
    \begin{align}
    \Viab_{[0,\tau]}(\K,\U) \supseteq \Viab_{[0,\tau]}(\Pi_1 \K, \U) &\times \Inv_{[0,\tau]}(\Pi_2 \K,\V(\cdot))_*\notag, \\
        \text{where} \quad \V \colon [0,\tau] &\to 2^{\Sl_1}; \;\,
                    t \mapsto \Viab_{[0,\tau-t]}(\Pi_1 \K, \U). \label{E:incl_unicouple_1}
    \end{align}
\end{thm}

\begin{proof}
    We first show that the inclusion holds for any set $\D \subset \Sl_1$ in which $x_1$ takes value. Since both inputs (control $u$ and ``disturbance'' $v:=x_1 \in \D$) are disjoint across the two subsystems we have 
    \begin{align}
         \Viab_{[0,\tau]}(\Pi_1 \K,\U) \times \Inv_{[0,\tau]}(\Pi_2 \K,\D)_* &= \Disc_{[0,\tau]}(\Pi_1 \K,\U,\{0\})_* \times \Disc_{[0,\tau]}(\Pi_2 \K,\{0\},\D)_*\notag \\ &\overset{\text{Thm\ref{Thm:Viab_equals_projectionViabs}}}{=} \Disc_{[0,\tau]}(\K,\U,\D)_*. \label{E:inclusion_crossprod1}
    \end{align}
    With $\D=\V(\cdot)$, inclusion \eqref{E:incl_unicouple_1} follows from Lemma~\ref{Lem:disc_subset_viab}:
    \begin{equation*}
        \Viab_{[0,\tau]}(\Pi_1 \K,\U) \times \Inv_{[0,\tau]}(\Pi_2 \K,\V(\cdot))_* = \Disc_{[0,\tau]}(\K,\U,\V(\cdot))_* \subseteq \Viab_{[0,\tau]}(\K,\U).
    \end{equation*}

    Note that the set-valued map $\V(\cdot)$ at time $t$ is the finite-horizon viability kernel of the upper subsystem over the interval $[0,\tau-t]$. This map is continuous (it is both lower and upper semicontinuous (cf.\ \cite{aubin1991viability}) at every point in its domain) and non-decreasing \cite{cardaliaguet1999set} (i.e.\ $\V(t) \supseteq \V(s)$ $\forall t \in [s,\tau]$, $ s\in [0,\tau]$), with $\Pi_1\K$ being its upper-limit in the sense of Kuratowski (Definition~\ref{Def:kuratowski_up_lim}) as $t\to \tau^-$ and $\Viab_{[0,\tau]}(\Pi_1 \K, \U)$ its lower-limit as $t\to 0^+$. Furthermore, since $\Pi_1\K$ and $\U$ are convex and compact and the dynamics linear, the sets $\V(t)$ are also convex and compact at every $t$. From this we have that $\Inv_{[0,\tau-s]}(\Pi_2 \K, \V(s))$ is continuous, convex and compact for every $s$, and non-decreasing over $s\in [0,\tau]$ \cite{aubin1991viability}.

    We use these statements to argue that a digression from the formulation in \eqref{E:incl_unicouple_1} loses its sufficiency to guarantee an under-approximation in the sense that if the uncertainty set is assumed to be a subset of $\V(t)$ for any $t$ then the cross-product may not generate an under-approximation of the viability kernel: Consider a set-valued map $\widetilde{\V}(\cdot)$ s.t.\ $\exists \hat{t}\in [0,\tau]$, $\widetilde{\V}(\hat{t}) \subseteq \V(\hat{t})$ (e.g.\ a constant set $\Viab_{[0,\tau]}(\Pi_1 \K, \U)$ $\forall t$). It is clear from \eqref{E:inclusion_crossprod1} that
    \begin{align}\label{E:superset_if_otherV2}
         \Viab_{[0,\tau]}(\Pi_1 \K,\U) \times \Inv_{[0,\tau]}(\Pi_2 \K,\widetilde{\V}(\cdot))_* \supseteq \Disc_{[0,\tau]}(\K,\U,\V(\cdot))_*
    \end{align}
    since for any set $\C$, $\Inv_{[0,\tau-\hat{t}]}(\C,\V(\hat{t}))_* \subseteq \Inv_{[0,\tau-\hat{t}]}(\C,\widetilde{\V}(\hat{t}))_*$ and therefore $\Inv_{[0,\tau-s]}(\C,\V(\cdot))_* \subseteq \Inv_{[0,\tau-s]}(\C,\widetilde{\V}(\cdot))_*$ $\forall s \in [0,\hat{t}]$. There is no guarantee that this superset in \eqref{E:superset_if_otherV2} is a subset of $\Viab_{[0,\tau]}(\K,\U)$; Lemma~\ref{Lem:disc_subset_viab} is no longer applicable. On the flip side, if $\widetilde{\V}(\cdot)$ is such that $\widetilde{\V}(t) \supseteq \V(t)$ for any $t \in [0,\tau]$ (e.g.\ a constant set $\Pi_1 \K$ $\forall t$), then an excessively conservative under-approximation is obtained.
\end{proof}

\subsection{Sub-Interval Formulation and Decentralized Algorithm}

In practice, we can perform the analysis over sub-intervals (similarly to \cite{G05}) while still maintaining conservatism. During each sub-interval the set $\V(\cdot)$ is sampled and kept constant in \emph{backward time}. Such sub-interval analysis is possible via the semi-group property in both subspaces as well as the following results in $\Sl_2$.


\begin{prop}\label{Prop:sub_interval_analysis}
    For $N:=\tau/q$, $N\in \mathbb{N}$ time steps each of length $q\in \Real^+$ we have that
    \begin{equation}
        \bigcap_{i=0}^{N-1} \C_i \subseteq \Inv_{[0,\tau]}(\Pi_2 \K,\V(\cdot))_*
    \end{equation}
    %
    where $\C_i = \Inv_{[0,q]}(\C_{i+1},\V((i+1)q))_*$ with $\C_N=\Pi_2 \K$.
\end{prop}

\begin{proof}
    Notice that since $\{\V(t)\}_{t=0}^{\tau}$ is a non-decreasing sequence of compact and convex sets with $\V(t)\subset \Sl_1 =: \Real^{p_v}$ we have that for a fixed $q$, for every $i$, $\V((i+1)q)\supseteq \V(t)$ $\forall t\in [0,(i+1)q]$. Using this, the fact that $\C_i \subseteq \C_{i+1} \subseteq \C_N$ $\forall i$, and the semi-group property we have
    \begin{align*}
        x_0 \in \bigcap\nolimits_{i=0}^{N-1} \C_i
        &\Leftrightarrow \forall i \in [0,N-1], \, \forall v_i(\cdot)\in \{v_i\colon [0,q] \to \Real^{p_v} \, \text{measurable},\\ &\qquad \qquad \, v_i(s) \in \V((i+1)q) \; \text{a.e.\ $s\in [0,q]$}  \},\, \forall t \in [0, q], \,  \hat{\xi}^2_{x_0,v_i}(t) \in \C_{i+1}\\
        &\Rightarrow \forall i \in [0,N-1], \, \forall v_i(\cdot)\in \{v_i\colon [iq,(i+1)q] \to \Real^{p_v} \, \text{measurable},\\ &\qquad \qquad \, v_i(s) \in \V(s) \; \text{a.e.\ $s\in [iq,(i+1)q]$}  \},\, \forall t \in [iq, (i+1)q], \,  \hat{\xi}^2_{x_0,v_i}(t) \in \C_{i+1}\\
        %
        %
        %
        &\Rightarrow \forall v(\cdot)\in \{v\colon [0,\tau] \to \Real^{p_v} \, \text{measurable}, \, v(t) \in \V(t) \, \text{a.e.} \}, \, \forall t \in [0,\tau], \, \hat{\xi}^2_{x_0,v}(t) \in \C_N\\
        &\Rightarrow x_0 \in \Inv_{[0,\tau]}(\Pi_2\K,\V(\cdot))_*,
    \end{align*}
    where $v$ is the concatenation of functions $v_i$ over $[0,\tau]$.

\end{proof}

In the limit this set converges to the invariance kernel with unsampled input set.

\begin{defn}[Kuratowski upper and lower limits \cite{cardaliaguet1999set}]\label{Def:kuratowski_up_lim}
    Let $\{\A(s)\}_{s\in S}$ be a sequence of subsets in a metric space $(E,d)$. The upper-limit of $\A(s)$ as $s \to \hat{s}$ is
    \begin{equation*}
        \Limsup_{s\to \hat{s}} \A(s) := \left\{ x \in E \mid \liminf_{s\to \hat{s}} d(x,\A(s)) = 0 \right\},
    \end{equation*}
    where $d(x,\A):= \inf_{a \in \A} d(x,a)$. Its lower-limit is
    \begin{equation*}
        \Liminf_{s\to \hat{s}} \A(s) := \left\{ x \in E \mid \lim_{s\to \hat{s}} d(x,\A(s)) = 0 \right\}.
    \end{equation*}
\end{defn}

\begin{prop}\label{Prop:lim_of_intersectionSi}
    Denote by $\C_\cap(q):= \bigcap_{i=0}^{N-1} \C_i$ the intersection of $N=\tau/q$ sub-interval invariance kernels from Proposition~\ref{Prop:sub_interval_analysis}. For the sequence of subsets $\{\C_\cap(q)\}_{q\geq 0}$ we have
    \begin{equation}\label{E:lim_of_intersectionSi}
        \Limsup_{q\to 0^+} \C_\cap(q) = \Inv_{[0,\tau]}(\Pi_2\K,\V(\cdot))_*.
    \end{equation}
\end{prop}

\begin{proof}
    Given $q$, define a piecewise constant set-valued map $\V_{\text{sh}}(t;q) := \V(iq)$ $\forall t$ for which $i$ is the unique integer in $\{1,\dots,N\}$ satisfying $t \in ((i-1)q,iq]$ when $t$ varies backwards from $\tau$ to $0$ (i.e.\ a backward sample and hold of $\V(\cdot)$). Recall that $\V(\cdot)$ is non-decreasing and continuous, and $\V(t)$ compact for every $t$. Clearly, $\V_{\text{sh}}(\cdot;q) \supseteq \V(\cdot)$ $\forall q$. The sequence $\{\V_{\text{sh}}(\cdot;q)\}_{q\geq 0}$ converges to $\V(\cdot)$ from outside: We say that $\tilde{v}(\cdot;q) \in \V_{\text{sh}}(\cdot;q)$ iff $\tilde{v}(t;q) \in \V_{\text{sh}}(t;q)$ $\forall t$. As $q \to 0^+$, $\forall \tilde{v}(\cdot;q)\in \V_{\text{sh}}(\cdot;q)$ $\forall \epsilon \geq 0$ $\forall t$ $\B(\tilde{v}(t;q), \epsilon) \cap \V(t) \neq \emptyset$, where $\B(x,\epsilon)$ denotes the ball (associated with a metric $d$) of radius $\epsilon$ centered at $x$. In other words, $\forall \tilde{v}(\cdot;q)\in \V_{\text{sh}}(\cdot;q)$, $\exists v(\cdot) \in \V(\cdot)$ s.t.\ $\limsup_{q\to 0^+} d(v(\cdot),\tilde{v}(\cdot;q))= \liminf_{q\to 0^+} d(v(\cdot),\tilde{v}(\cdot;q))=0$. So $\lim_{q\to 0^+} d(v(\cdot),\V_{\text{sh}}(\cdot;q))=0$, and therefore $\Liminf_{q\to 0^+} \V_{\text{sh}}(\cdot;q)=\V(\cdot)$. On the other hand, we know from the semi-group property that $\C_\cap(q) = \Inv_{[0,\tau]}(\Pi_2\K,\V_{\text{sh}}(\cdot;q))_*$. Hence,
    \begin{equation*}
        \Limsup_{q\to 0^+} \C_\cap(q) = \Limsup_{q\to 0^+} \Inv_{[0,\tau]}(\Pi_2\K,\V_{\text{sh}}(\cdot;q))_* = \Inv_{[0,\tau]}(\Pi_2\K,\Liminf_{q\to 0^+} \V_{\text{sh}}(\cdot;q))_* = \Inv_{[0,\tau]}(\Pi_2\K, \V(\cdot))_*.
    \end{equation*}
\end{proof}

Using this formulation we can perform the decentralized analysis in Theorem~\ref{Thm:Viab_supset_projectionViabs} via Algorithm~\ref{Alg:subIntervalAnalysis} over sub-intervals.
\begin{alg}
    \begin{algorithmic}[1]
    \State $N \gets \tau/q$ \Comment{Assumed integer.}
    \State $\C_N \gets \Pi_2 \K$
    \State $\V_N \gets \Pi_1 \K$
    \For{$i = N-1$ to $0$}
        \State $\C_i \gets \Inv_{[0,q]}(\C_{i+1},\V_{i+1})_*$ \label{Alg_Step:inv}
        \State $\V_i \gets \Viab_{[0,q]}(\V_{i+1},\U)$
    \EndFor
    \State \textbf{return} $\V_0 \times \C_0$ \Comment{$\subseteq \Viab_{[0,\tau]}(\K,\U)$}
    \end{algorithmic}
    \caption{Sub-Interval Decentralized Computations}
    \label{Alg:subIntervalAnalysis}
\end{alg}



\subsection{Bounding the Approximation in $\Sl_2$}

Notice from Theorem~\ref{Thm:Viab_supset_projectionViabs} that the computed construct in the upper subspace is exact in that
\begin{equation}
    \Pi_1 \Viab_{[0,\tau]}(\K,\U) = \Viab_{[0,\tau]}(\Pi_1 \K, \U).
\end{equation}
On the other hand additional conservatism is introduced in the lower subspace $\Sl_2$ due to treating the effect of the upper subsystem as a worst-case disturbance. Quantifying this error remains an open problem. However, we can formulate a qualitative lower bound on the shrinkage of the invariance kernel in $\Sl_2$ in backward time. This bound will be expressed in terms of system-specific (and ultimately, design-specific) parameters that form the desired structure \eqref{E:sys_unicoupled_DisjointB}:


Following \cite{kurzhanskii96ellipsoidal}, the invariance kernel in $\Sl_2$ can be expressed as
\begin{equation}
    \Inv_{[0,\tau]}(\Pi_2 \K,\V(\cdot))_* = \bigcap_{t\in[0,\tau]} \left( e^{-t A_2} \Pi_2 \K  \ominus  \int_0^t e^{-rA_2} \Delta \V(t-r) dr  \right)
\end{equation}
with $\ominus$ denoting the Pontryagin difference. Let $\B(\delta)$ be the norm-ball of radius $\delta \in \Real^+$ about the origin, and define $\eta \colon \Real^+ \to \Real^+$,
\begin{equation}
    \eta(s) := \frac{e^{s\norm{A_2}}-1}{\norm{A_2}}.
\end{equation}
Bounding the contribution of the uncertainty (disturbance) in computation of the invariance kernel over the interval $[0,\theta]$ we have \cite{G05} that
\begin{align}
    \int_0^\theta e^{-rA_2} \Delta \V(\theta-r) dr
        &\subseteq \B\left(\norm{\int_0^\theta e^{-rA_2} \Delta \V(\theta-r) dr}\right)\\
        &\subseteq  \B\left(\int_0^\theta e^{r\norm{A_2}} \norm{\Delta} \sup_{x\in \V(\theta-r)} \norm{x} dr\right)\\
        &\subseteq  \B\left(\norm{\Delta} \sup_{x\in \V(\theta)} \norm{x} \int_0^\theta e^{r\norm{A_2}} dr\right)\\
        &\subseteq  \B\left( \norm{\Delta} \sup_{x\in \V(\theta)} \norm{x}  \eta(\theta) \right).
\end{align}
Clearly, this contribution is weakened as $\norm{\Delta} \to 0$. Further, we have
\begin{equation}\label{E:inclusion_normed}
    \bigcap_{i=0}^{N-1} \left( \bigcap_{t\in[0,q]}  e^{-t A_2} \C_{i+1}   \ominus  \B\left( \norm{\Delta} \sup_{x\in \V((i+1)q)} \norm{x} \eta(q) \right)\right) \subseteq  \bigcap_{i=0}^{N-1} \Inv_{[0,q]}(\C_{i+1},\V((i+1)q))_*
\end{equation}
with $\C_N := \Pi_2 \K$. From the dual of the results in \cite{G05}, we know that the Hausdorff distance of the two sets in the inclusion above decreases as $q \to 0^+$, and tends to zero if $\V(iq) = \B(\sup_{x\in \V(iq)} \norm{x})$. The Kuratowski upper-limit of the left-hand-side of \eqref{E:inclusion_normed} is therefore $\Inv_{[0,\tau]}(\Pi_2 \K,\V(\cdot))_*$ as $q\to 0^+$ (via Proposition~\ref{Prop:lim_of_intersectionSi}). Now, notice that for sufficiently small $q\ll 1$,
\begin{equation}
    \eta(q) = \lim_{M\to \infty} \sum_{j=1}^{M} \frac{q^j  (\norm{A_2})^{j-1} }{j!} \leq \lim_{M\to \infty} \sum_{j=1}^{M} \frac{q^j  (\overline{\sigma}(A_2) \sqrt{\tilde{n}})^{j-1} }{j!} = q + \frac{q^2}{2} \overline{\sigma}(A_2) \sqrt{\tilde{n}}+O(q^3),
\end{equation}
where $\overline{\sigma}(A_2)$ and $\tilde{n}=\operatorname{dim}(\Sl_2)$ respectively denote the largest singular value and the dimension of the lower subsystem. Therefore \eqref{E:inclusion_normed} provides a qualitative lower-bound on how much $\Inv_{[0,\tau]}(\Pi_2 \K,\V(\cdot))_*$  can shrink in backward time in terms of $\tilde{n}$, the magnitude of the unidirectional coupling $\norm{\Delta}$, the supremum of $\V(t)$ (the viability kernel in $\Sl_1$), and the largest singular value $\overline{\sigma}(A_2)$ of the lower subsystem. If we can choose $\tilde{n}$ appropriately, assign the slow eigenvalues to the lower subsystem, and weaken the effect of the disturbance (uncertainty) as much as possible by minimizing $\norm{\Delta}$, we can expect the conservatism to be reduced considerably. The proposed modified Riccati transformation in Section~\ref{S:Riccati} provides this flexibility while imposing the desired structure \eqref{E:sys_unicoupled_DisjointB} on the system.

\subsection{Decentralized Viability in Transformed Coordinates}

Suppose that for a general system \eqref{E:general_ss_eqn} under which a centralized viability computation is known to be burdensome, there exists an invertible transformation $z = T^{-1}x$ such that in the new coordinates the system $\widetilde{\s} = T^{-1}(\s)$ has the form of either \eqref{E:sys_decoupled_DisjointB} or \eqref{E:sys_unicoupled_DisjointB}. Suppose that Assumption~\ref{Ass:K_axis-aligned} is satisfied for $T^{-1}\K$. When the transformation yields decoupled $A$-matrix as well as disjoint input, Theorem~\ref{Thm:Viab_equals_projectionViabs} under the transformed dynamics $\widetilde{\s}$ becomes:
%
%
\begin{cor}\label{Cor:Viab_equals_projectionViabs_Riccati}
    $\Viab_{[0,\tau]}(\K,\U) = T\Viab_{[0,\tau]}^{\widetilde{\s}}(T^{-1}\K,\U) =  T\left(\Viab_{[0,\tau]}^{\widetilde{\s}}(\Pi_1 T^{-1} \K, \U_1) \times \Viab_{[0,\tau]}^{\widetilde{\s}}(\Pi_2 T^{-1} \K,\U_2)\right)$, where the superscript $\widetilde{S}$ is used to specify when a construct is formed under the transformed dynamics.
\end{cor}

For the more general case Theorem~\ref{Thm:Viab_supset_projectionViabs} implies:
%
%
\begin{cor}\label{Cor:Viab_supset_projectionViabs_Riccati}
    $\Viab_{[0,\tau]}(\K,\U) \supseteq T \left(\Viab_{[0,\tau]}^{\widetilde{\s}}(\Pi_1 T^{-1} \K, \U) \times \Inv_{[0,\tau]}^{\widetilde{\s}}(\Pi_2 T^{-1} \K,\V(\cdot))_* \right)$ with $\V(t) := \Viab_{[0,\tau-t]}^{\widetilde{\s}}(\Pi_1 T^{-1} \K, \U)$ $\forall t\in [0,\tau]$.
\end{cor}

Decentralized analysis over sub-intervals are performed similarly to Algorithm~\ref{Alg:subIntervalAnalysis}, and a lower-bound for the shrinkage of the invariance kernel in $\Sl_2$ can be formulated according to \eqref{E:inclusion_normed} with $\C_N = \Pi_2 T^{-1}\K$ and the respective transformed system matrices. Note that in $\Sl_1$, $\Pi_1 T^{-1}\Viab_{[0,\tau]}(\K,\U) = \Viab_{[0,\tau]}^{\widetilde{\s}}(\Pi_1 T^{-1} \K, \U)$, and that the computed construct in $\Sl_2$ is a guaranteed under-approximation of the projection of the actual viability kernel in that subspace, i.e.\ $\Inv_{[0,\tau]}^{\widetilde{\s}}(\Pi_2 T^{-1} \K,\V(\cdot))_* \subseteq \Pi_2 T^{-1}\Viab_{[0,\tau]}(\K,\U)$. We present one such transformation next.

\section{The Riccati-Based Transformation}\label{S:Riccati}

We draw upon the so-called Riccati transformation---a two-stage coordinate transformation based on the solutions of a nonsymmetric algebraic Riccati equation (NARE) and a Sylvester equation. This transformation, originally introduced in \cite{Ch72} for decoupling of singularly perturbed systems, was later generalized in \cite{Kok75} to larger classes of \emph{autonomous} LTI systems. An in-depth overview of the application of this transformation in optimal control theory, singular perturbation theory, and asymptotic approximation theory can be found in \cite{Smith1987}, while more recent advances are given in \cite{Gajic2000, Shim2006}.

Let \eqref{E:general_LTI_system} be partitioned as
\begin{equation}\label{E:general_partitioned}
    \s = \begin{bmatrix}
        \begin{array}{cc|c}
           A_{11} & A_{12} &  B_1\\
           A_{21} & A_{22} &  B_2\\
        \end{array}\\
        \end{bmatrix}
\end{equation}
with $A_{11} \in \Real^{k \times k}$, $A_{12} \in \Real^{k \times (n-k)}$, $A_{21} \in \Real^{(n-k) \times k}$, $A_{22} \in \Real^{(n-k) \times (n-k)}$, $B_{1} \in \Real^{k \times p}$, and $B_{2} \in \Real^{(n-k) \times p}$, for some $k<n$. Now consider the nonsingular transformation matrices
\begin{align}
  T_1 &= \begin{bmatrix}
        I_{k} & \mathbf{0}\\
        -L &        I_{n-k}
      \end{bmatrix} \in \Real^{n \times n}, \label{E:T1_riccati} \\
  T_2 &= \begin{bmatrix}
    I_{k} & M \\
    \mathbf{0} &  I_{n-k}
  \end{bmatrix}\in \Real^{n \times n}, \label{E:T2_riccati}
\end{align}
where $I_n$ denotes the $n\times n$ identity matrix. With $L \in \Real^{(n-k) \times k}$ and $M \in \Real^{k \times (n-k)}$ that satisfy
\begin{align}
  &\text{(NARE:)}   &\mathscr{R}(L) &:= L A_{11} - A_{22}L - L A_{12}L + A_{21} = \mathbf{0}, \label{E:ARE0} \\
  &\text{(Sylvester:)} &\mathscr{S}(M) &:= \bigl( A_{11}-A_{12}L \bigr) M - M \bigl(A_{22}+L A_{12} \bigr) + A_{12} = \mathbf{0}, \label{E:Sylv0}
\end{align}
the transformed system is
{\allowdisplaybreaks
\begin{align}
\s'  = T_1^{-1}(\s)
            &=\begin{bmatrix}
            \begin{array}{cc|c}
               A_{11} - A_{12}L  & A_{12} &  B_1\\
               \cancelto{\mathbf{0}}{\mathscr{R}(L)}   & A_{22}+L A_{12} &  L B_1 + B_2\\
            \end{array}
            \end{bmatrix}\label{E:T1_in_standard},\\
\s'' = T_2^{-1}(\s')
            &=\begin{bmatrix}
            \begin{array}{cc|c}
               A_{11} - A_{12}L  & \cancelto{\mathbf{0}}{\mathscr{S}(M)} &  (I-ML) B_1 - M B_2\\
               \mathbf{0}   & A_{22}+L A_{12} &  L B_1 + B_2\\
            \end{array}
            \end{bmatrix}.
\end{align}
}%
Solutions to \eqref{E:ARE0} and \eqref{E:Sylv0} may not always exist. The above procedure is referred to as the (standard) Riccati transformation. If the control input is disjoint across the subsystems of $\s''$ (and thus the transformation imposes a structure similar to \eqref{E:sys_decoupled_DisjointB}), Corollary~\ref{Cor:Viab_equals_projectionViabs_Riccati} can be employed to approximate the viability kernel in a decentralized fashion based on subsystem analysis.

\subsection{The Modified Riccati Transformation}

For the more general case, on the other hand, we propose the following transformation that imposes a structure given in \eqref{E:sys_unicoupled_DisjointB} which also relaxes the condition on the shape of the set $\U$. Corollary~\ref{Cor:Viab_supset_projectionViabs_Riccati} can thus be employed to compute a conservative approximation of the true viability kernel.

\subsubsection{Transformation 1 (ETUC Subsystem)}

Consider a transformation through which the lower subsystem can be made ETUC. That is, in \eqref{E:T1_in_standard} for the transformation matrix $T_1$ we seek an $L$ in $\mathscr{R}(L)$ that is also a solution of $L B_1 + B_2 = \mathbf{0}$.

\begin{ass}\label{A:riccati_LB1=B2_solvability}
   $\mathscr{C}(\tr{B_2}) \subseteq \mathscr{C}(\tr{B_1})$, where $\mathscr{C}(X)$ is the column-space of matrix $X$.
\end{ass}



\begin{lem}[\cite{Rao1972,GroB1999}]\label{L:AX=B}
  Under Assumption~\ref{A:riccati_LB1=B2_solvability} the class of solutions of $L B_1 = -B_2$ w.r.t.\ $L\in \Real^{(n-k)\times k}$ can be characterized by
  \begin{equation}\label{E:AX=B}
        \mathcal{L} := \left\{ -B_2 B_1^\dagger + Z - Z B_1 B_1^\dagger, \;\; Z \in \Real^{(n-k) \times k} \right\}
  \end{equation}
  with $\dagger$ denoting the Moore-Penrose pseudoinverse.
\end{lem}

Assumption~\ref{A:riccati_LB1=B2_solvability} is the necessary and sufficient condition for solvability of $L B_1 = -B_2$. Substituting \eqref{E:AX=B} for $L$ in $\mathscr{R}(L)$ we obtain
\begin{equation}\label{E:R_(37)}
    \begin{split}
      \widehat{\mathscr{R}}(Z):= Z\Xi + \Gamma & + Z  \Bigl(A_{12}-  B_1 B_1^\dagger A_{12} \Bigr) Z (B_1 B_1^\dagger - I)\\
      &\qquad + \Bigl(A_{22} - B_2 B_1^\dagger A_{12} \Bigr) Z (B_1 B_1^\dagger - I),
    \end{split}
\end{equation}
where
\begin{align}
    \Xi & = -(B_1 B_1^\dagger - I) \Bigl( A_{11} + A_{12} B_2 B_1^\dagger \Bigr), \\
    \Gamma &= \Bigl( A_{22} B_2 B_1^\dagger +A_{21} \Bigr) - B_2 B_1^\dagger \Bigl(A_{12} B_2 B_1^\dagger + A_{11} \Bigr). \label{E:riccati_Gamma_defn}
\end{align}

To eliminate the non-invertible term $(B_1 B_1^\dagger - I)$ from the right-hand side of \eqref{E:R_(37)} we equate $\widehat{\mathscr{R}}(Z)$ to some rank correcting term $\delta \mathscr{F}(Z)$ with
\begin{equation}\label{E:riccati_F(Z)}
    \mathscr{F}(Z) := Z  \Bigl( A_{12}-B_1 B_1^\dagger A_{12} \Bigr) Z + \Bigl(A_{22} - B_2 B_1^\dagger A_{12} \Bigr) Z
\end{equation}
and $\delta \in \Real \backslash \{-1, 0\}$ a finite (but possibly large) parameter such that $\big(B_1 B_1^\dagger - (\delta+1)I \big)$ is nonsingular:
\begin{align}
      \widehat{\mathscr{R}}(Z) &= Z\Xi + \Gamma + Z  \Bigl(A_{12}-  B_1 B_1^\dagger A_{12} \Bigr) Z (B_1 B_1^\dagger - I) \notag \\
      &\qquad\qquad\qquad\qquad + \Bigl(A_{22} - B_2 B_1^\dagger A_{12} \Bigr) Z (B_1 B_1^\dagger - I)\\
      & = Z\Xi + \Gamma + \mathscr{F}(Z) (B_1 B_1^\dagger - I) \doteq \delta \mathscr{F}(Z). \label{E:ARE0_subs(37)}
\end{align}
Simple algebraic manipulation and post-multiplication of $\widehat{\mathscr{R}}(Z) - \delta \mathscr{F}(Z) = \mathbf{0}$ by $\big(B_1 B_1^\dagger - (\delta+1)I \big)^{-1}$ results in a NARE in the variable $Z$:
\begin{equation}\label{E:ARE1}
  \mathscr{R}_1(Z) := Z \tilde{A}_{11} - \tilde{A}_{22} Z - Z \tilde{A}_{12} Z + \tilde{A}_{21} = \mathbf{0}
\end{equation}
with $\tilde{A}_{11} = \Xi \, \big( B_1 B_1^\dagger - (\delta+1)I \big)^{-1}$, $\tilde{A}_{21}= \Gamma \, \big(B_1 B_1^\dagger - (\delta+1)I \big)^{-1}$, $\tilde{A}_{12} = \big( B_1 B_1^\dagger A_{12}-A_{12} \big)$, and $\tilde{A}_{22}= \big( B_2 B_1^\dagger A_{12} - A_{22} \big)$.
\begin{prop}\label{P:both_eqns_simultaneous}
    If a root $Z \in \Real^{(n-k) \times k}$ of the NARE \eqref{E:ARE1} exists, it constitutes an $L \in \mathcal{L}$ that simultaneously satisfies
    %
    \begin{subequations}\label{E:both_equations_simlutaneous}
      \begin{align}
        & L B_1 + B_2 = \mathbf{0}, \\
        & \mathscr{R}(L) = L A_{11} - A_{22}L - L A_{12}L + A_{21} = \delta \mathscr{F}(Z).
      \end{align}
    \end{subequations}
\end{prop}

\begin{proof}
    By virtue of \eqref{E:ARE0_subs(37)}, a matrix $Z$ that satisfies \eqref{E:ARE1} also satisfies \eqref{E:both_equations_simlutaneous} via \eqref{E:AX=B}.
\end{proof}

\begin{rem}
    If $p\ge k$ the set $\mathcal{L}$ reduces to the singleton $\{-B_2 B_1^\dagger\}$ and the method still applies.
\end{rem}

\begin{thm}\label{T:Transform1_riccati}
  The transformation \eqref{E:T1_riccati} with $L \in \Real^{(n-k) \times k}$ obtained through Proposition~\ref{P:both_eqns_simultaneous} makes the lower subsystem in \eqref{E:general_partitioned} ETUC. Moreover, the coupling terms are altered such that the effect of the upper subsystem on the evolution of the lower subsystem is parameterized by $\delta$.
\end{thm}

\begin{proof}
    %
    \begin{align}
    \s^{\prime}  = T_1^{-1}(\s)
                &=\begin{bmatrix}
                \begin{array}{cc|c}
                   A_{11} - A_{12}L  & A_{12} &  B_1\\
                   L A_{11} - A_{22}L - L A_{12}L + A_{21}    & A_{22}+L A_{12} &  L B_1 + B_2\\
                \end{array}\\
                \end{bmatrix}\\
                &=\begin{bmatrix}
                \begin{array}{cc|c}
                   A_{11} - A_{12}L  & A_{12}  &  B_1\\
                   \delta \mathscr{F}(Z)   & A_{22}+L A_{12} &  \mathbf{0}\\
                \end{array}
                \end{bmatrix}. \label{E:G-prime_modified}
    \end{align}
\end{proof}

\begin{rem}
    Note that the imposed $\delta$-parameterization of the off-diagonal term $\delta\mathscr{F}(Z)$ in \eqref{E:G-prime_modified} provides an additional degree of freedom in adjusting (minimizing) the coupling of the two subsystems in the new coordinates. This will be discussed further in Section~\ref{S:choosing_delta}.
\end{rem}

Nonsymmetric Riccati equations have long been an active area of research \cite{Freiling2002}. To solve \eqref{E:ARE1} we draw on the fixed-point algorithm described in \cite{Kok75} and derive the necessary conditions for the existence and uniqueness of a real root $Z$. Suppose $\bigl( B_2 B_1^\dagger A_{12} - A_{22} \bigr)$ is invertible. Define initial values as
\begin{align}
    Z_0 & := \bigl( B_2 B_1^\dagger A_{12} - A_{22} \bigr)^{-1} \Gamma \bigl(B_1 B_1^\dagger - (\delta+1)I \bigr)^{-1}, \label{E:riccati_Z0}\\
    A_0 & := \Xi \bigl(B_1 B_1^\dagger - (\delta+1)I \bigr)^{-1} - \bigl( B_1 B_1^\dagger A_{12} - A_{12} \bigr) Z_0 .
\end{align}
To find $Z$ we look for
\begin{equation}\label{E:D=Z-Z0}
    D := Z - Z_0
\end{equation}
by solving
\begin{equation}\label{E:ARE_in_D_recast}
    \begin{split}
      \widetilde{\mathscr{R}_1}(D) := D A_0 - \Bigl( B_2 B_1^\dagger A_{12} - &A_{22} + Z_0 \bigl( B_1 B_1^\dagger A_{12} - A_{12} \bigr) \Bigr) D \\
      & - D \bigl( B_1 B_1^\dagger A_{12} - A_{12} \bigr) D + Z_0 A_0 = \mathbf{0}.
    \end{split}
\end{equation}

\begin{lem}[{\cite[Lem.\ 1]{Kok75}}] \label{L:fixed_point_algorithm1}
    Suppose $\bigl( B_2 B_1^\dagger A_{12} - A_{22} \bigr)$ is nonsingular. If
    \begin{equation}\label{E:riccati_alg_cond1}
        \bigl\lVert \bigl( B_2 B_1^\dagger A_{12} - A_{22} \bigr)^{-1}\bigr\rVert \leq \frac{1}{3 \Bigl( \norm{A_0} + \norm{B_1 B_1^\dagger A_{12} - A_{12}} \norm{Z_0} \Bigr)}
    \end{equation}
    then \eqref{E:ARE_in_D_recast} has a unique real root $D$ that satisfies
    \begin{equation}\label{E:riccati_D_upperbound}
        0 \leq \norm{D} \leq \frac{2 \norm{A_0} \norm{Z_0}}{\norm{A_0} + \norm{B_1 B_1^\dagger A_{12} - A_{12}} \norm{Z_0}}
    \end{equation}
    and is the fixed-point solution of the contraction $D_{k+1} = \mathscr{P}_1(D_k)$ given by
    \begin{equation}\label{E:riccati_alg_contraction1}
        \begin{split}
            \mathscr{P}_1(D_k) := \bigl( B_2 B_1^\dagger & A_{12} - A_{22} \bigr)^{-1} \Bigl( Z_0 A_0 + D_k A_0 \\
            & - Z_0 \bigl( B_1 B_1^\dagger A_{12} - A_{12} \bigr) D_k - D_k \bigl(B_1 B_1^\dagger A_{12} - A_{12} \bigr) D_k \Bigr).
        \end{split}
    \end{equation}
\end{lem}


\begin{rem}
  As in \cite{Kok75} it can be shown that the relative error $e_k := \norm{D_k - D}/\norm{D}$ after $k$ iterations is bounded above by
  \begin{equation}\label{E:e_k_forARE(D)}
    e_k \leq \biggl( 3 \bigl\lVert\bigl( B_2 B_1^\dagger A_{12} - A_{22} \bigr)^{-1}\bigr\rVert \Bigl( \norm{A_0} + \norm{B_1 B_1^\dagger A_{12} - A_{12}} \norm{Z_0} \Bigr) \biggr)^{\!\!k}
  \end{equation}
  and decreases as $\abs{\delta}$ increases since $\norm{A_0}$ and $\norm{Z_0}$ are inversely related to $\abs{\delta}$. 
\end{rem}

For a given $\delta$, using $D_0 = \mathbf{0}$ as initial condition we compute $D$ iteratively. The fixed-point solution $D^* = \mathscr{P}_1(D^*)$ is then used to obtain $Z = D^* + Z_0$ which in turn solves $\mathscr{R}_1(Z)=\mathbf{0}$ in \eqref{E:ARE1} and results in a matrix $L$, through \eqref{E:AX=B}, that satisfies both equations in \eqref{E:both_equations_simlutaneous}.

\subsubsection{Transformation 2 (Unidirectionally Coupled Subsystems)}

Consider the NARE
\begin{equation}\label{E:ARE2}
    \mathscr{R}_2(M) = \bigl( A_{11} - A_{12} L \bigr) M - M \bigl( A_{22} + L A_{12} \bigr) - M \bigl( \delta \mathscr{F}(Z) \bigr) M + A_{12} = \mathbf{0}.
\end{equation}
For a given $L$, $\delta$, and $Z$, if there exists a solution $M$ that satisfies \eqref{E:ARE2}, we obtain the following:

\begin{thm}\label{T:Transform2_riccati}
    The transformation \eqref{E:T2_riccati} with $M \in \Real^{k \times (n-k)}$ satisfying NARE \eqref{E:ARE2} makes the subsystems in \eqref{E:G-prime_modified} unidirectionally coupled.
\end{thm}

\begin{proof}
    \begin{align}
        \s'' = T_2^{-1}(\s') &= \begin{bmatrix}
            \begin{array}{cc|c}
               A_{11} - A_{12}L - M \delta \mathscr{F}(Z)  & \cancelto{\mathbf{0}}{\mathscr{R}_2(M)} &  B_1\\
               \delta \mathscr{F}(Z)   & A_{22}+L A_{12} + \delta \mathscr{F}(Z) M &  \mathbf{0}\\
            \end{array}
            \end{bmatrix}.
            \label{E:riccati_final_transformed}
    \end{align}
\end{proof}

\begin{rem}
    In the transformed coordinates the lower subsystem remains ETUC. Furthermore, the $\delta$-parameterization of the unidirectional coupling between subsystems is also preserved.
\end{rem}

Before further analyzing the unidirectional coupling term $\delta \mathscr{F}(Z)$, let us derive the necessary conditions for the existence and uniqueness of a solution $M$ to \eqref{E:ARE2} to be used with the same convergent iterative procedure described previously. For a given $\delta$, $Z$, and $L$, let $\bigl( A_{11} - A_{12} L \bigr)$ be invertible and the initial values be defined as
\begin{align}
    M_0 & := - \bigl( A_{11} - A_{12} L \bigr)^{-1} A_{12},\\
    N_0 & := A_{22} + L A_{12} + \delta \mathscr{F}(Z) M_0.
\end{align}
We seek $M$ by forming
\begin{equation}\label{E:J=M-M0}
    J:= M - M_0
\end{equation}
and solving
\begin{equation}\label{E:ARE_in_S_recast}
    \widetilde{\mathscr{R}_2}(J) := J N_0 - \bigl( A_{11} - A_{12} L - \delta M_0 \mathscr{F}(Z) \bigr) J + \delta J \mathscr{F}(Z) J + M_0 N_0 = \mathbf{0}.
\end{equation}

\begin{lem}[{\cite[Lem.\ 1]{Kok75}}] \label{L:fixed_point_algorithm2}
    Suppose $\bigl( A_{11} - A_{12} L \bigr)$ is nonsingular. If
    \begin{equation}\label{E:riccati_alg_cond2}
        \bigl\lVert \bigl( A_{11} - A_{12} L \bigr)^{-1}\bigr\rVert \leq \frac{1}{3 \Bigl( \norm{N_0} + \norm{\delta \mathscr{F}(Z)} \norm{M_0} \Bigr)}
    \end{equation}
    then \eqref{E:ARE_in_S_recast} has a unique real root $J$ that satisfies
    \begin{equation}\label{E:riccati_J_upperbound}
        0 \leq \norm{J} \leq \frac{2 \norm{N_0} \norm{M_0}}{\norm{N_0} +  \norm{\delta \mathscr{F}(Z)} \norm{M_0}}
    \end{equation}
    and is the fixed-point solution of the contraction $J_{k+1} = \mathscr{P}_2(J_k)$ given by
    \begin{equation}\label{E:riccati_alg_contraction2}
            \mathscr{P}_2(J_k) := \bigl( A_{11} - A_{12} L \bigr)^{-1} \Bigl( M_0 N_0 + J_k N_0 + \delta M_0 \mathscr{F}(Z) J_k + \delta J_k \mathscr{F}(Z) J_k \Bigr).
    \end{equation}
\end{lem}

\begin{rem}
    The relative error $e_k := \norm{J_k - J}/\norm{J}$ after $k$ iterations is bounded above by
    \begin{equation}\label{E:e_k_forARE(J)}
    e_k \leq \biggl( 3 \bigl\lVert\bigl( A_{11} - A_{12} L \bigr)^{-1}\bigr\rVert \Bigl( \norm{N_0} + \norm{\delta \mathscr{F}(Z)} \norm{M_0} \Bigr) \biggr)^{\!\!k}
    \end{equation}
    and decreases as $\norm{\delta \mathscr{F}(Z)}$, $\norm{A_{22}}$, and $\bigl\lVert \bigl( A_{11} - A_{12} L \bigr)^{-1}\bigr\rVert$ decrease. This occurs when the ill-conditioning of the $A$\nobreakdash-matrix increases (e.g.\ in the case of two-time-scale systems; see \cite{kokotovic1999singular} and the references therein) and $\delta$ is chosen such that $\norm{\delta \mathscr{F}(Z)}$ is minimized. 
\end{rem}

Using $J_0 = \mathbf{0}$ as initial condition we compute $J$ iteratively. The fixed-point solution $J^* = \mathscr{P}_2(J^*)$ is then used to obtain $M = J^* + M_0$ which in turn solves $\mathscr{R}_2(M)=\mathbf{0}$ in \eqref{E:ARE2}.


Note that both conditions \eqref{E:riccati_alg_cond1} and \eqref{E:riccati_alg_cond2} are conservative and their satisfaction ensures rapid convergence (usually within 2 or 3 iterations). In practice, the right-hand-side of these inequalities can be relaxed up to 10 times in most cases without causing divergence. 

\subsubsection{The Unidirectional Coupling Term (Choosing $\delta$)} \label{S:choosing_delta}

Finally, we analyze the unidirectional coupling term $\delta \mathscr{F}(Z)$ and its behavior with respect to the free parameter $\delta$. Since $Z$ is an implicit function of $\delta$, we adopt the extended notation $\delta \mathscr{F}(Z(\delta))$ to reflect this dependency. First, we formalize a conservative upper-bound on $\norm{\delta \mathscr{F}(Z(\delta))}$ as an explicit function of $\delta$. This assures that the unidirectional coupling remains bounded for almost all admissible values of the free parameter $\delta$.

\begin{prop}\label{T:riccati_conserv_bound}
    The worst-case unidirectional coupling between the two subsystems in the transformed coordinates, i.e. $\norm{\delta \mathscr{F}(Z(\delta))}$ in \eqref{E:riccati_final_transformed}, is (conservatively) bounded above such that
    \begin{equation}\label{E:riccati_dF(Z)_upperbound}
        \norm{\delta \mathscr{F}(Z(\delta))} \leq  \frac{1}{\abs{\delta}} \biggl( \frac{\abs{\delta}+1}{\abs{\delta + 1}} \biggr)^{\!\! 2} a + \biggl( \frac{\abs{\delta}+1}{\abs{\delta + 1}} \biggr) b,  \quad \forall \delta \in \Real \backslash \{-1, 0\},
    \end{equation}
    where the constants $a$ and $b$ are independent of $\delta$ and are determined by $a :=  \alpha (b/\beta)^2$, $b := 3 \normshort{B_1 B_1^\dagger} \gamma \beta$, $\gamma := \norm{\Gamma} \normshort{\bigl( A_{22} - B_2 B_1^\dagger A_{12} \bigr)^{-1}}$, $\alpha := \normshort{A_{12} - B_1 B_1^\dagger A_{12}}$, and $\beta := \normshort{A_{22} - B_2 B_1^\dagger A_{12}}$.
\end{prop}
\begin{proof}
    The proof is provided in the Appendix.
\end{proof}

%

Now consider inequalities \eqref{E:riccati_alg_cond1} and \eqref{E:riccati_alg_cond2}, which are dependant on $\delta$. Adequately chosen and sufficiently large values of $\delta$ help ensure that these conditions are met. On the other hand, choosing $\delta$ exceedingly large defeats the purpose of $\delta$-parameterization of the unidirectional coupling term, since it can be shown that as $\delta$ grows, $\norm{\delta \mathscr{F}(Z(\delta))}$ approaches a problem-dependant constant that may not necessarily be an extremum point.

\begin{prop}\label{P:lim_f(delta)}
    $\displaystyle \lim_{\delta \to \pm \infty} \norm{\delta \mathscr{F}(Z(\delta))} = \norm{\Gamma}$ with $\Gamma$ given by \eqref{E:riccati_Gamma_defn}.
\end{prop}

\begin{proof}
    This proof is also provided in the Appendix.
\end{proof}

It follows from Proposition~\ref{P:lim_f(delta)} that $0 \leq \inf_{\delta} \norm{\delta \mathscr{F}(Z(\delta))} \leq \norm{\Gamma}$. Therefore naively letting $\abs{\delta} \to \infty$ essentially removes the added flexibility associated with the $\delta$\nobreakdash-parameterization in the modified Riccati approach and instead enforces a trivial solution $L = -B_2 B_1^\dagger$. While for some systems this solution may yield the smallest possible unidirectional coupling between the resulting subsystems (i.e.\ a unidirectional coupling with the least infinity norm), in most cases a carefully chosen $\delta$ not only facilitates the satisfaction of the convergence conditions \eqref{E:riccati_alg_cond1} and \eqref{E:riccati_alg_cond2}, but also further minimizes the worst-case unidirectional coupling. Thus, formulated as an optimization problem, we seek a $\delta$ that solves the following:
\begin{align}
    &\minimize_{\delta \in \Real \backslash \set{-1,0}} \quad f(\delta) := \norm{\delta \mathscr{F}(Z(\delta))} \notag \\
    &\; \st  \qquad \text{\eqref{E:riccati_alg_cond1} and \eqref{E:riccati_alg_cond2}}. \notag
\end{align}
Note that this is a nonconvex problem, and in general, $f$ may be a non-smooth function of $\delta$. However, a global optimum need not be computed. Any suboptimal solution can be used as long as that solution yields a satisfactory degree of unidirectional coupling between the subsystems in the transformed coordinates. An approximation to the optimum point can be obtained numerically, for example by fine-griding the real line or using the bisection algorithm.


In practice, while the exact shape of the function $f$ is problem-dependant, we have found (but not proven) that in most cases it exhibits a behavior similar to that of an absolute value proper rational function (over a discontinuous domain) of the form
\begin{equation}
    \hat{f}(\delta) = \Bigl\lvert \frac{c_0}{\delta^k} + c_1\Bigr\rvert + c_2, \quad \forall \delta \in  \mathcal{D},
\end{equation}
where $\mathcal{D} \subset \Real \backslash \set{-1,0}$ is the union of the two segments of the real line for which the magnitude of $\delta$ is large enough such that \eqref{E:riccati_alg_cond1} and \eqref{E:riccati_alg_cond2} are both satisfied, $k \in \mathbb{N}$, $k: \text{odd}$, $c_0 = -c_1 (\delta^*)^k$, $\delta^* = \arg \min_{\delta \in \mathcal{Y}} f(\delta)$, $c_2 = \min_{\delta \in \mathcal{Y}} f(\delta)$, and $c_1 = \bigl( \lim_{\delta \to \pm \infty} f(\delta) \bigr) - c_2 = \norm{\Gamma} - c_2$. 

\begin{ex}\label{Ex:Example_arbitrary4D}
    Consider the system
    {\footnotesize
        \begin{displaymath}
            A = \begin{bmatrix}
                \p1.5072  & \p3.3984  &  \p0.1300  & -0.0884\\
                \p5.0644  & -2.6683  &  \p0.0227  &  \p0.1689\\
                \p0.1156  & -0.1863  &  \p0.5686  &  \p0.2648\\
                -0.0808  & \p0.0229  &  \p0.4915  & \p0.5949\\
                 \end{bmatrix}, \quad
            B = \begin{bmatrix}
               -0.7433\\
               -2.2528\\
               -0.9075\\
               \p0.6036\\
            \end{bmatrix}.
        \end{displaymath}
    }
    Fig.~\ref{F:riccati_coupling} shows $f(\delta)$ and its approximation $\hat{f}(\delta) = \abs{-\frac{27.65}{\delta} + 0.55} + 1.82$ evaluated where \eqref{E:riccati_alg_cond1} and \eqref{E:riccati_alg_cond2} hold.
\end{ex}

\begin{figure}[t]
\centering
\includegraphics[scale=0.47]{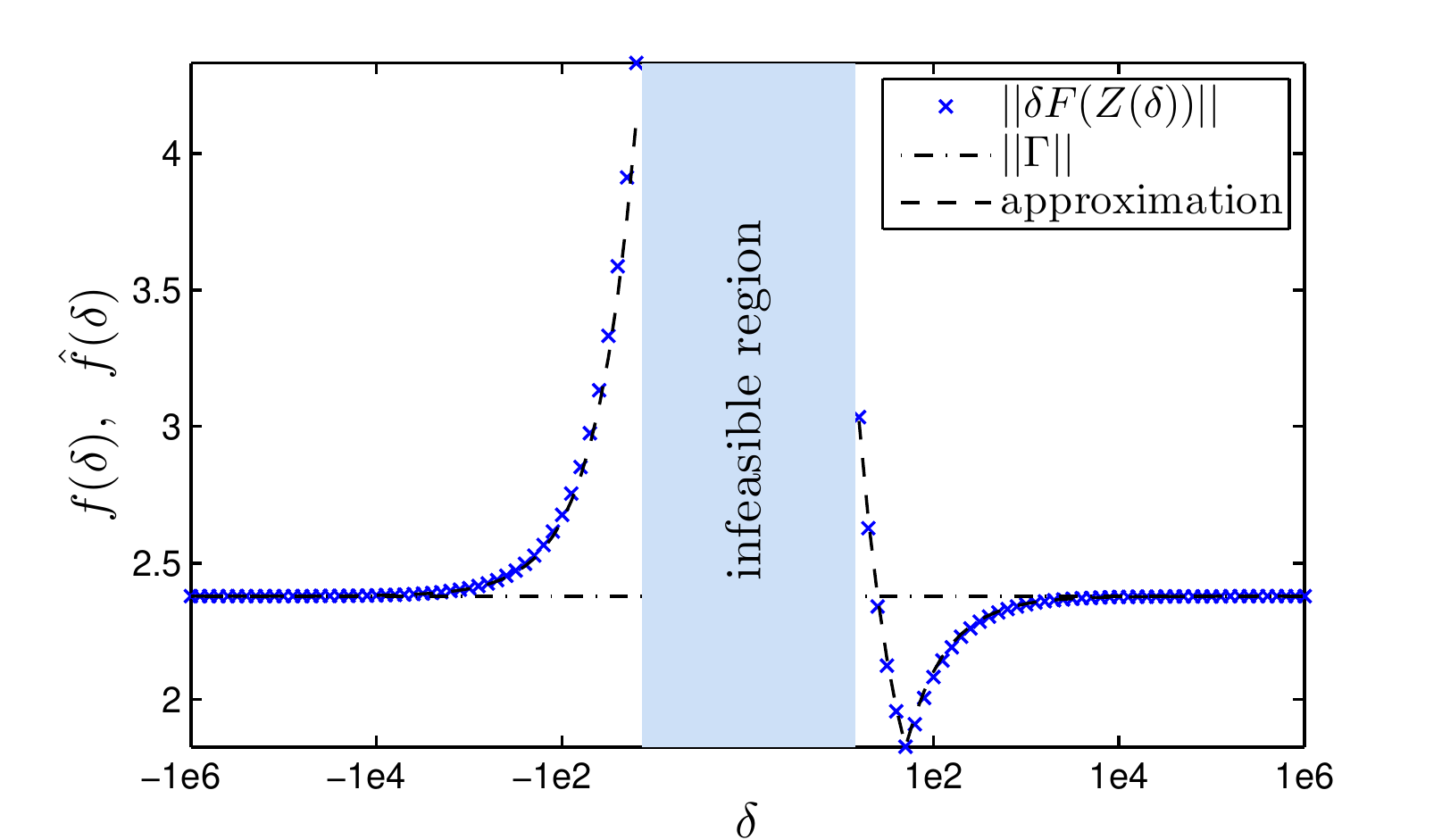}
\caption{The worst-case unidirectional coupling $f(\delta) = \norm{\delta \mathscr{F}(Z(\delta))}$ ($\times$'s) and its approximation $\hat{f}(\delta) = \abs{\frac{-27.65}{\delta} + 0.55} + 1.82$ (dashed) computed for Example \ref{Ex:Example_arbitrary4D}. The interval $(-15, +15)$ over which \eqref{E:riccati_alg_cond1} and \eqref{E:riccati_alg_cond2} are violated is labeled as ``infeasible region''. The asymptote $\lim_{\delta \to \pm \infty} f(\delta) = \norm{\Gamma}$ (dash-dotted) is also shown. The minimum of $f(\delta)$ occurs when $\delta \approx +50$.}
\label{F:riccati_coupling}
\end{figure}

A randomized, empirical test in \cite[Section~4.4.2]{kaynama2012Thesis} examines the potential affect of the system dimension $n$ on the magnitude of the unidirectional coupling and the amount of time consumed by the decomposition process. While the test shows an increasing trend in average values, there is significant variance. In addition, the time required for the decomposition process (even for the highest dimension $n=16$ in our test) is still negligible (${\sim}1.5\,\text{s}$) compared to the time required for the actual viability computations.

\subsection{Recursive Decomposition}\label{S:recursive_decomp}

A recursive decomposition when the standard Riccati transformation can be used is straightforward. Suppose that the modified Riccati transformation is used throughout the process. In deeper level recursions, the decomposition can be applied to the uppermost subsystem since that subsystem is controlled whereas every other subsystem is ETUC. For example, to decompose a 6D system into three 2D subsystems, in the first recursion level, the partitioning can be chosen such that the resulting upper (controlled) subsystem is 4D and the lower (ETUC) subsystem is 2D. In the second recursion level, if the solutions exist, the 4D subsystem is then decomposed into two 2D subsystems. Note that in the recursive application of the decomposition, when the modified Riccati transformation is employed, all subsystems but one are ETUC. Therefore, this iterated decomposition may result in an excessively conservative under-approximation of the true viability kernel.

\subsection{Riccati-Based Viability in Lower Dimensions}

In the new coordinates $z = T^{-1}x$, $T=T_1T_2$, the subsystem dynamics are governed by
\begin{align}
    \dot{z}_1 &= \bigl( A_{11} - A_{12}L - \delta M \mathscr{F}(Z) \bigr) z_1 + B_1 u, \label{E:decomposed_subsystem1}\\
    \dot{z}_2 &= \bigl( A_{22} + LA_{12} + \delta \mathscr{F}(Z)M \bigr) z_2 + B_2 u + \delta \mathscr{F}(Z) z_1 \label{E:decomposed_subsystem2}
\end{align}
with $\delta \mathscr{F}(Z) = \mathbf{0}$ when the standard Riccati transformation yields disjoint input, or $B_2 = \mathbf{0}$ when the modified Riccati transformation is employed. In the latter case, $\delta=\delta^*$ is precomputed so as to minimize $\norm{\delta \mathscr{F}(Z)}$. In addition, the transformation automatically assigns the slowest eigenvalues to the lower subsystem. These in turn prevent excessive conservatism in approximation of the construct in $\Sl_2$. Analysis over sub-intervals are performed according to Algorithm~\ref{Alg:subIntervalAnalysis}, and a qualitative lower-bound for the shrinkage of the invariance kernel in $\Sl_2$ can be formulated according to \eqref{E:inclusion_normed} with $\C_N = \Pi_2 T^{-1}\K$ and $\Delta = \delta \mathscr{F}(Z)$.

\section{Numerical Examples}\label{S:Riccati_examples}

Among Eulerian methods we use the Level Set Toolbox (LS) v.1.1 \cite{MitchellHSCC05} for our analysis. All computations are performed on a dual core Intel-based machine with $2.8\,\text{GHz}$ CPU, $6\,\text{MB}$ L2 cache and $3\,\text{GB}$ RAM running single-threaded 32-bit \textsc{Matlab} 7.5.

\subsection{4D Cart with Two Inverted Pendulums}\label{Example:riccati_cart}

Consider the linearized model of a cart with two separately mounted inverted pendulums from \cite[Ex.\ 2.2.1]{ioannou1996robust} with $l_1 = 30$, $l_2 = 35$:
\[
    A = \begin{bmatrix}
         0   & 1       &  0      &   0\\
    0.3920   &      0  & -0.0327   &      0\\
         0   &      0  &       0   & 1\\
    0.0560   &      0  &  \p0.2753  &       0\\
         \end{bmatrix}, \quad
    B = \begin{bmatrix}
        0\\
        -0.0033\\
        0\\
        -0.0005\\
    \end{bmatrix}.
\]
The state vector $x\in \Real^4$ consists of angular displacement of each inverted pendulum from vertical and the corresponding angular velocities; The input $u \in \Real$, $\abs{u} \leq 10$, arises from a force applied to the cart.

Note that despite the sparsity of the system no permutation matrix can recover our desired structures \eqref{E:sys_decoupled_DisjointB} or \eqref{E:sys_unicoupled_DisjointB} (the graph representation of this system is a strongly connected digraph). We decompose this system using the presented Riccati-based technique into two 2D subsystems, with unidirectional coupling determined by the solution $L = -B_2B_1^\dagger$ regardless of the value of $\delta$:
\[
    A'' = \begin{bmatrix}
        0  &  0.9524    &     0    &      0\\
        0.3920   &      0    &     0    &      0\\
         0  &  0.1429    &     0    & 1.0500\\
         0    &     0   & 0.2800     &     0\\
        \end{bmatrix}, \quad
    B'' = \begin{bmatrix}
            0 \\
       -0.0033 \\
            0 \\
            0 \\
    \end{bmatrix}.
\]

We choose $\K$ such that in the transformed coordinates we have the constraint set $\K_z := \{ z \mid \norm{z} \leq 0.5, \, z = T^{-1}x, \, x\in \K \}$. We seek to identify the set of initial states for which there exists a bounded control law that keeps the angular displacement of the pendulums contained in $\K_z$ and thus within a ball of finite radius about their upright positions, despite control saturation. We perform the analysis over $50$ sub-intervals. LS v.1.1 only accepts hyper-rectangular input sets. To comply with this limitation we modify Step~\ref{Alg_Step:inv} in Algorithm~\ref{Alg:subIntervalAnalysis} so that $\C_i \gets \Inv_{[0,q]}(\C_{i+1},\Boxed(\V_{i+1}))_*$, where $\Boxed(\A)$ is the \emph{interval hull} of $\A$. Conservatism in Proposition~\ref{Prop:sub_interval_analysis} is preserved since $\Boxed(\V(iq)) \supseteq \V(iq)$. Computations are performed over a grid with $41$ nodes in each dimension using a first-order accuracy for $\tau =3 \, \text{s}$ (Fig.~\ref{F:riccati_reach4D_II}). The computation time for the actual and the transformation-based kernels were $1098.48 \, \text{s}$ and $4.27 \, \text{s}$, respectively. The Riccati-based kernel covers $74\%$ of the volume of the full-order set (calculated based on the number of grids contained in each set).


\begin{figure}[t]
  \centering
      \includegraphics[scale=0.6]{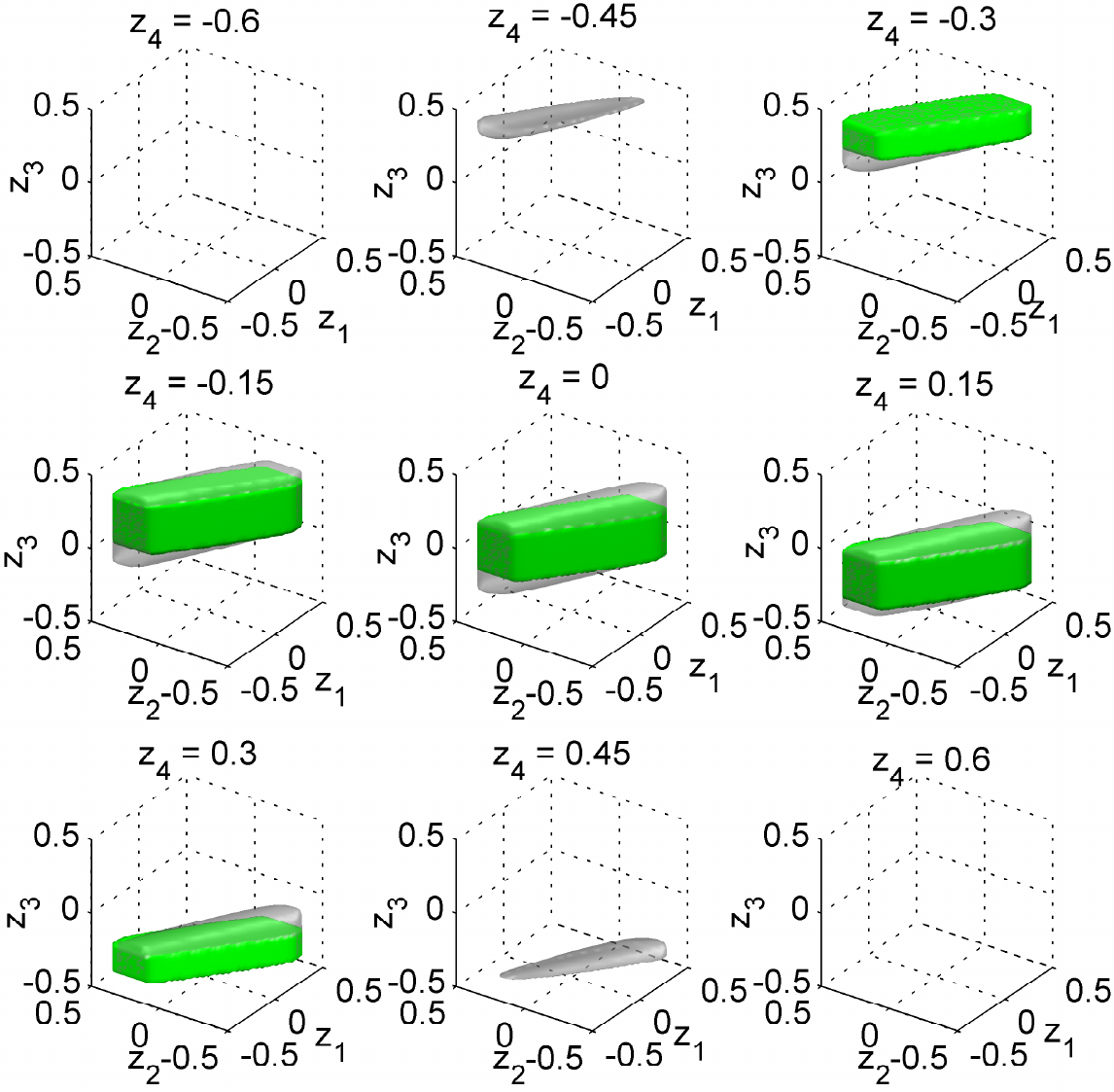}
      \caption{Riccati-based (solid, dark) vs.\ actual (transparent, light) viability kernels in the transformed coordinate space for Example~\ref{Example:riccati_cart}.}
      \label{F:riccati_reach4D_II}
\end{figure}

\subsection{Arbitrary 6D System}\label{Example:arbitrary_6D}

Consider the two-time-scale system $\dot{x} = \left[\begin{smallmatrix}
  A_{11} & A_{12}\\
  \epsilon A_{21} & \epsilon A_{22}
\end{smallmatrix}\right] x + \left[\begin{smallmatrix}
  B_1\\
  \epsilon B_2
\end{smallmatrix}\right] u$ with $\epsilon = 0.1$, and $A\in \Real^{6\times 6}$ and $B\in \Real^{6\times 2}$ matrices randomly drawn from a normal distribution $\mathcal{N}(0,1)$: 

{\footnotesize
\begin{displaymath}
    A = \begin{bmatrix}
       -0.3557  & -0.3078 &  -0.6097  &  \p2.0275 &  -1.3636  & -0.4131\\
        \p0.1233 &  -1.6441 &   \p0.2404 &  -0.6431  &  \p0.0517 &  -0.1454\\
        \p1.8857 &  -1.1748 &  -1.2502 &  -0.7252 &  -0.7801  & -0.3972\\
       -0.0194  & -0.0779  & -0.0208 &   \p0.0160 &  -0.0465  &  \p0.0535\\
       -0.0486  & -0.0192  &  \p0.0781  &  \p0.1017  &  \p0.0838  & -0.0518\\
        \p0.0043  & -0.0849 &  -0.0228  & -0.0901 &  -0.0319 &  -0.1143\\
         \end{bmatrix}, \quad
    B = \begin{bmatrix}
        \p1.0720 &  -0.8153\\
       -1.7390  & -0.7181\\
       -0.8292  & -0.4906\\
        \p0.0156 &   \p0.0540\\
       -0.0960   & \p0.0875\\
       -0.0347  & -0.0054\\
    \end{bmatrix}.
\end{displaymath}
}

We decompose this system into two 3D subsystems using the modified Riccati transformation with $\delta^* \approx -25$:

{\footnotesize
\begin{displaymath}
    A'' = \begin{bmatrix}
           -0.3472 &  -0.1553 &  -0.5243 &        0  &       0  &       0\\
           \p0.1252 &  -1.6394 &   \p0.2499 &        0  &       0  &       0\\
            \p1.8832 &  -0.9445 &  -1.1162 &        0  &       0  &       0\\
            \p0.0069 &  -0.1476 &  -0.0544 &  -0.1011  &  \p0.0244  &  \p0.1152\\
           -0.0523 &  -0.0749 &  -0.0097 &   \p0.1474  &  \p0.0156  & -0.0571\\
           -0.0015 &  -0.0604 &  -0.0238 &  -0.1425  &  \p0.0200  & -0.0762\\
         \end{bmatrix}, \quad
    B'' = \begin{bmatrix}
    \p1.0720  & -0.8153\\
       -1.7390 &  -0.7181\\
       -0.8292 &  -0.4906\\
             0 &        0\\
             0 &        0\\
             0 &        0\\
    \end{bmatrix}.
\end{displaymath}
}

The constraint $\K$ is chosen such that this set in the new coordinates is a nonconvex set formed by the cross-product of the union of a sphere and a hyper-rectangle as shown in Fig.~\ref{F:riccati_reach6D}. We choose $\U$ such that $-0.5 \leq u_1 \leq 0.5$ and $0.5 \leq u_2 \leq 1$. (The shape of $\U$ need not be rectangular since one of the subsystems is ETUC.)
Decentralized approximation of $\Viab_{[0,2]}(\K,\U)$ are carried out over $50$ sub-intervals using $151$ nodes in each dimension and a second-order accuracy (Fig.~\ref{F:riccati_reach6D}). The overall computation time was $1 \, \text{h}$ (including calculation of $\delta^*$, transformation matrices, the decomposition, and projections which took only a few seconds). In contrast, the actual kernel is prohibitively computationally expensive to compute with LS for any meaningful grid resolution. Moreover,  on average $350 \, \text{MB}$ of RAM was used in the Riccati-based viability calculations (of which $110\,\text{MB}$ was to store the grid), whereas the computation of the full-order kernel would require about $380 \,\text{TB}$ (terabyte) merely to store the grid.

\begin{figure}[t]
\centering
    \includegraphics[scale=0.4]{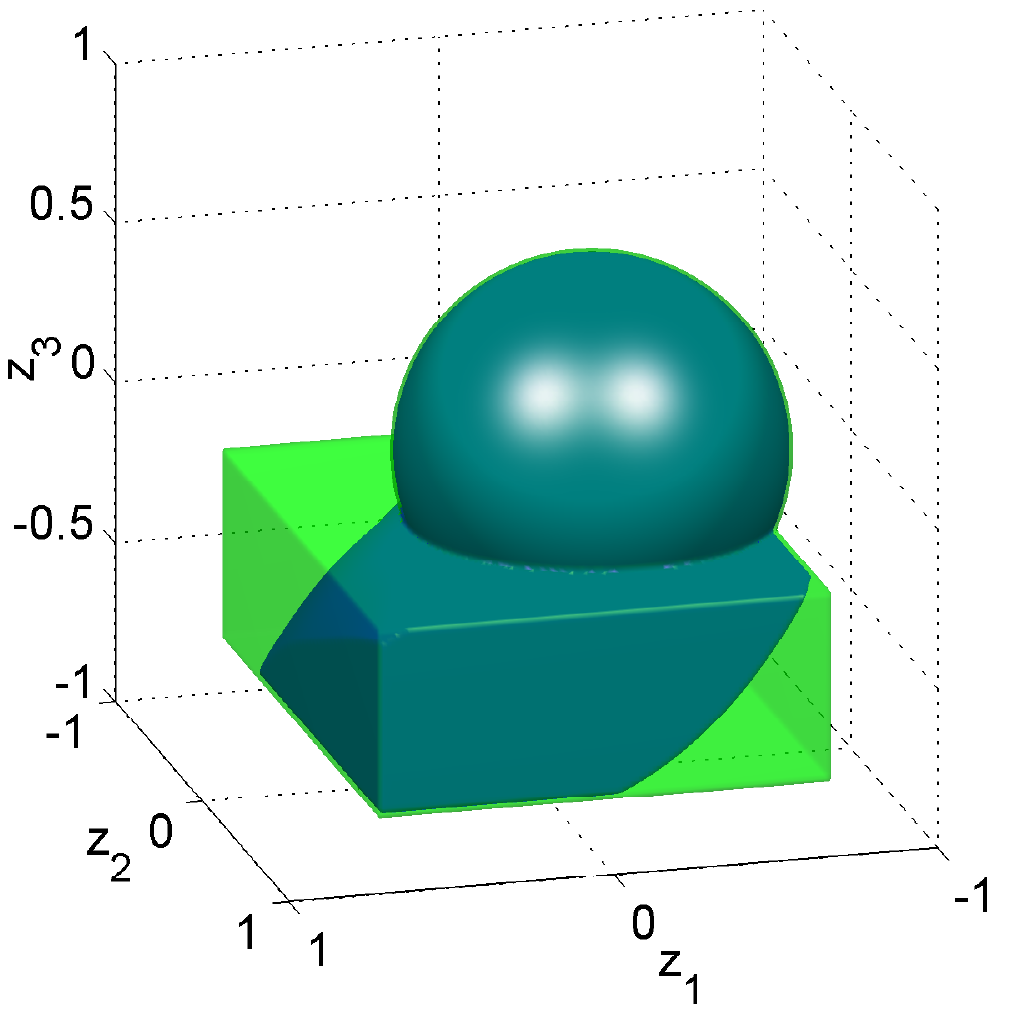}\hspace{20pt}
    \includegraphics[scale=0.4]{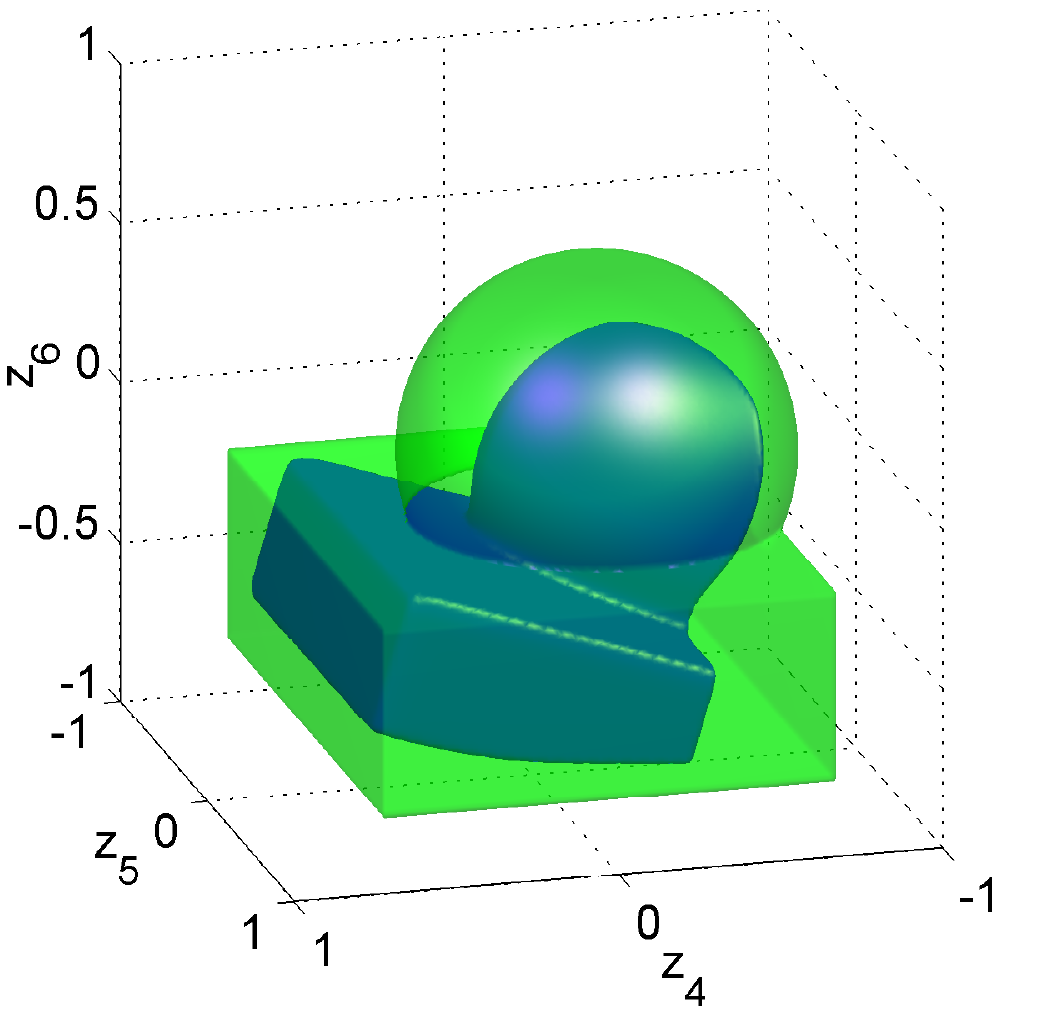}
  \caption{The constraint set (transparent) and its Riccati-based viability kernel in 3D subspaces of the transformed coordinates for Example~\ref{S:Riccati_examples}.}
  \label{F:riccati_reach6D}
\end{figure}

\subsection{Comparison With Schur-Based Decomposition (\cite{kaynama2011complexity})}

In \cite{kaynama2011complexity} we presented a Schur-based decomposition technique that is applicable to almost any LTI system. In contrast, the decomposition method presented here is based on two nonsymmetric algebraic Riccati equations. The existence of solutions to these
algebraic equations, however, is limited by a number of conditions on system matrices and is
therefore heavily problem dependent. Indeed, as pointed out earlier, the conditions are more
likely to be satisfied as the ill-conditioning of the original system matrices increases---e.g., for
two-time-scale systems.\footnote{cf.\ \cite[Figure~4.6]{kaynama2012Thesis} for the fraction of tests on randomly generated systems for which a solution existed.} However, when the algebraic Riccati equations do converge, the resulting subsystems could \emph{potentially} yield less conservative kernel approximations than in the case of the Schur-based decomposition; See \cite[Section~4.5.4]{kaynama2012Thesis} for a simple example. In general, however, it is the problem under study that determines which decomposition method is more suitable. A better strategy may be to use both decomposition techniques if possible and take the union of their resulting sets to obtain a more accurate under-approximation of the viability kernel than what could be achieved using each individual technique.



\section{Conclusions and Future Work}\label{S:Conc}

We considered the problem of guaranteed safety and constraint satisfaction in moderately-dimensioned, safety-critical LTI systems with compact, simply-connected state constraints. To provide such guarantees the computation of the viability kernel is required. Historically, the algorithms that approximate this set---known as Eulerian methods---are based on gridding the state space. While powerful and versatile, their computational complexity increases exponentially with the dimension of the state which renders them impractical for systems of dimensions higher than three or four. We investigated conditions under which the viability kernel can be conservatively approximated in a decentralized fashion in lower-dimensional subspaces. We then presented a new similarity transformation that imposes such conditions on the system, thereby allowing us to employ Eulerian methods on higher-dimensional systems. The transformation is best suited to two-time-scale systems.

It is possible (although uncommon) that the transformation matrix can become poorly-conditioned due to pseudoinverses and numerical algorithms involved, resulting in the state constraint set in the transformed coordinates becoming too severely distorted under the linear map to be of any practical use. An upper-bound on the condition number in terms of the system matrices and the free parameter $\delta$ is provided in \cite[Appendix B.2]{kaynama2012Thesis}. We are currently investigating possible remedies that would ensure a well-conditioned transformation matrix.

With the particular system structure \eqref{E:sys_unicoupled_DisjointB} considered in this paper, the computations in the upper subspace are exact. On the other hand, the lower subspace computations are subject to accuracy loss since the formulated disturbance is assumed to play optimally at all times, aiming to shrink the construct in that subspace. While this is to ensure that we obtain a conservative approximation, in reality it is quite likely that the input is not always adversarial. Moreover, here we have only required the disturbance signal be measurable, and thus it can vary discontinuously. We know, however, that the trajectories of the upper subsystem are continuous in time. Restricting the disturbance input to draw from the subclass of continuous signals may result in a more accurate approximation in the lower subspace. In either case, quantifying the accuracy loss in Lemma~\ref{Lem:disc_subset_viab} is an open problem. Another future direction is in investigating alternative system structures to the ones considered in Section~\ref{S:suitable_structures}.


\section*{Appendix} \label{appendix}


\begin{proof}[Proof of Proposition~\ref{T:riccati_conserv_bound}]
    From the matrix inversion lemma, $(Y + UCV)^{-1} = Y^{-1} - Y^{-1} U (C^{-1}+ VY^{-1} U )^{-1} VY^{-1}$, with $Y = -(\delta +1) I$, $U = B_1$, $C = I$, and $V = B_1^\dagger$ we have
    \begin{equation}\label{E:riccati_Sherman-Morrison}
        \bigl( B_1 B_1^\dagger - (\delta +1) I \bigr)^{-1} = -\frac{1}{\delta + 1} \Bigl( I + \frac{1}{\delta} B_1 B_1^\dagger \Bigr).
    \end{equation}
    Using this, \eqref{E:riccati_F(Z)}, \eqref{E:riccati_Z0}, \eqref{E:D=Z-Z0}, \eqref{E:riccati_D_upperbound}, multiplicative and triangular inequalities, and $\normshort{B_1 B_1^\dagger} \ge 1$,
    {\allowdisplaybreaks
    \begin{align*}
        \norm{\delta \mathscr{F}(Z(\delta))}
        &\leq \abs{\delta} \bigl( \alpha (\norm{Z_0} + \norm{D})^2 + \beta (\norm{Z_0} + \norm{D}) \bigr)\\
        &\leq \abs{\delta} \Biggl( \alpha \biggl(\norm{Z_0} + \frac{2 \norm{A_0} \norm{Z_0}}{\norm{A_0} + \alpha \norm{Z_0}}\biggr)^{\!\! 2} + \beta \biggl(\norm{Z_0} + \frac{2 \norm{A_0} \norm{Z_0}}{\norm{A_0} + \alpha \norm{Z_0}}\biggr) \Biggr)\\
        &\leq \abs{\delta} \bigl( 9 \alpha \norm{Z_0}^2 + 3 \beta \norm{Z_0} \bigr)\\
        &\leq \abs{\delta} \Bigl( 9 \alpha \gamma^2 \bigl\lVert \bigl( B_1 B_1^\dagger - (\delta +1) I \bigr)^{-1}\bigr\rVert^2 + 3 \beta \gamma \bigl\lVert \bigl( B_1 B_1^\dagger - (\delta +1) I \bigr)^{-1}\bigr\rVert \Bigr)\\
        &\leq \abs{\delta} \biggl( 9 \alpha \gamma^2 \Bigl\lvert\frac{1}{\delta+1}\Bigr\rvert^2 \Bigl( 1 + \Bigl\lvert \frac{1}{\delta}\Bigr\rvert \Bigr)^2 \norm{B_1 B_1^\dagger}^2 + 3 \beta \gamma \Bigl\lvert\frac{1}{\delta+1}\Bigr\rvert \Bigl( 1 + \Bigl\lvert\frac{1}{\delta}\Bigr\rvert \Bigr) \norm{B_1 B_1^\dagger} \biggr)\\
        & \leq \frac{1}{\abs{\delta}} \biggl( \frac{\abs{\delta}+1}{\abs{\delta + 1}} \biggr)^{\!\! 2} a + \biggl( \frac{\abs{\delta}+1}{\abs{\delta + 1}} \biggr) b,   \quad \;\; \forall \delta \in \Real \backslash \{-1, 0\}.
     \end{align*}
     }
\end{proof}


\begin{proof}[Proof of Proposition~\ref{P:lim_f(delta)}]
    Notice from \eqref{E:riccati_alg_cond1} and \eqref{E:e_k_forARE(D)} that for large values of $\delta$, $Z$ can be closely approximated by its initial value $Z_0$. Using \eqref{E:riccati_Sherman-Morrison},
    \begin{equation*}
        \begin{split}
            \lim_{\delta \to \pm \infty} \norm{\delta \mathscr{F}(Z(\delta))} &= \lim_{\delta \to \pm \infty} \Bigl\lVert \frac{\delta}{(\delta + 1)^2} Q_1 (I + \frac{1}{\delta} B_1 B_1^\dagger) P_1 Q_1 (I + \frac{1}{\delta} B_1 B_1^\dagger)\\
            & \qquad \qquad  + \frac{\delta}{\delta + 1} P_2 Q_1 (I + \frac{1}{\delta} B_1 B_1^\dagger)\Bigr\rVert = \norm{0 + P_2 Q_1} = \norm{\Gamma}
        \end{split}
    \end{equation*}
    with $Q_1 := \bigl( B_2 B_1^\dagger A_{12} - A_{22} \bigr)^{-1} \Gamma$, $P_1 := (A_{12} - B_1 B_1^\dagger A_{12})$, $P_2 := (B_2 B_1^\dagger A_{12} - A_{22})$.
\end{proof}

\section*{Acknowledgment}

The authors thank I.\ Mitchell and R.\ Nagamune for valuable discussions, and the Associate Editor and anonymous reviewers for their constructive comments.

\bibliographystyle{IEEEtran}
\bibliography{IEEEfull,a_comp_reduct_journal09_v7}

\end{document}